\newif\iflipics
\lipicsfalse
\newif\ifllncs
\llncsfalse
\let\accentvec\vec  

\newif\ifhasappendix
\hasappendixfalse

\ifllncs

\documentclass[a4paper,english]{lipics}
\let\vec\accentvec 
\iflipics\else
\fi
\iflipics\else

\fi
\else
\documentclass[a4paper,11pt]{article}
\fi
\usepackage{etex,etoolbox}

\usepackage{amsmath}
\usepackage{amsfonts}
\usepackage{amssymb}
\usepackage{graphicx}
\usepackage{color}
\ifllncs \else
\usepackage[lofdepth,lotdepth]{subfig}
\fi
\usepackage{tikz}
\newcommand{\drawvertex}[1]{\draw[fill=white] #1 circle (0.12);}

\tikzstyle{vertex}=[circle,fill=white,draw,very thick, minimum size=7pt,inner sep=1pt]

\ifllncs
\iflipics
\newtheorem{Conjecture}[theorem]{Conjecture}
\else
\makeatletter
\renewcommand*\l@author[2]{}
\renewcommand*\l@title[2]{}
\makeatletter

\spnewtheorem*{Conjecture}{Conjecture}{\itshape\bfseries}{\rmfamily}
\spnewtheorem*{Proof}{Proof}{\itshape}{\rmfamily}
\renewenvironment{proof}{\begin{Proof}}{\qed\end{Proof}}
\fi
\else
\usepackage{fullpage}
\usepackage{amsthm}
\theoremstyle{plain}
\newtheorem{theorem}{Theorem}
\newtheorem{lemma}[theorem]{Lemma}

\newtheorem{corollary}[theorem]{Corollary}
\newtheorem{Conjecture}[theorem]{Conjecture}

\theoremstyle{definition}
\newtheorem{definition}{Definition}

\fi

\ifhasappendix
\usepackage{multibib}
\newcites{app}{References for the Appendix}
\fi

 \makeatletter
 \providecommand{\@fourthoffour}[4]{#4}  
 \ifllncs
 \providecommand{\@motsacces}[6]{#3} 
 \else
 \providecommand{\@motsacces}[4]{#4} 
 \fi 
 
 \def\fixstatement#1{%
   \AtEndEnvironment{#1}{%
  
        \xdef\pat@label{\expandafter\expandafter\expandafter
             \@motsacces\csname#1\endcsname\space~\@currentlabel}}}
 \globtoksblk\prooftoks{1000}
 \newcounter{proofcount}
 
 \long\def\proofatend#1\endproofatend{%
 \ifhasappendix
 \marginpar{\scriptsize{\vspace{-1cm}{\color{gray}{Proof in Appendix}}}}
   \edef\next{\noexpand \begin{proof}[Proof of \pat@label]
   }%
   \toks\numexpr\prooftoks+\value{proofcount}\relax=\expandafter{\next#1
   \end{proof}
   }
   \stepcounter{proofcount}
 \else
 \begin{proof}
 #1
 \end{proof}
 \fi
 }
 \long\def\differingappendixstatement#1#2{
 \ifhasappendix
 \edef\next{}%
   \toks\numexpr\prooftoks+\value{proofcount}\relax=\expandafter{\next#2}
   \stepcounter{proofcount}
  \else
  #1
 \fi
 }
 
\long\def\hereorinappendixstatement#1{%
\ifhasappendix%
\edef\next{}%
\toks\numexpr\prooftoks+\value{proofcount}\relax=\expandafter{\next#1}%
\stepcounter{proofcount}%
\else
#1
\fi
}

 \def\printproofs{%
   \count@=\z@
   \loop
     \the\toks\numexpr\prooftoks+\count@\relax
     \ifnum\count@<\value{proofcount}%
     \advance\count@\@ne
   \repeat}
 \makeatother
 \fixstatement{theorem}
 \fixstatement{lemma}
 \fixstatement{corollary}
 \fixstatement{definition}
 
\newcommand{\classC}{\mathcal{C}}

\newcommand{\dunion}{\mathbin{\dot{\cup}}}

\definecolor{darkred}{rgb}{0.5,0,0}
\definecolor{darkblue}{rgb}{0,0,0.5}
\definecolor{darkgreen}{rgb}{0,0.5,0}

\usepackage[normalem]{ulem}

\newcommand{\pascalsout}[1]{{\textcolor{blue}{\sout{[#1]}}}}

\ifllncs
\else
\usepackage[colorlinks]{hyperref}
\hypersetup{colorlinks
  ,linkcolor=darkred
  ,filecolor=darkgreen
  ,urlcolor=darkred
  ,citecolor=darkblue,
  linktocpage}
\fi

\usepackage{algorithm}
\usepackage{algorithmic}
\usepackage{setspace}
\newcommand{\ENSUREGAP}{\vspace{0.2cm}}
\definecolor{gray}{gray}{0.3}

\DeclareMathOperator{\ingredient}{Inv}
\DeclareMathOperator{\ingredientprime}{Inv_p}

\DeclareMathOperator{\ModDecFunc}{Mod}
\DeclareMathOperator{\id}{id}

\DeclareMathOperator{\Aut}{Aut}

\usepackage{xspace}

\newcommand{\forbind}[1]{(#1)\mbox{-\textup{\textsf{free}}}}
\newcommand{\forbindplain}[1]{(#1)\mbox{-{free}}}

\newcommand{\forbindONE}[1]{#1\mbox{-\textup{\textsf{free}}}}
\newcommand{\forbindONEplain}[1]{#1\mbox{-{free}}}

\newcommand{\HabcGraph}[3]{{H(#1,#2,#3)}}
\newcommand{\HabcGraphFOUR}[4]{{H(#1,#2,#3,#4)}}

\title{Towards an Isomorphism Dichotomy for Hereditary Graph Classes}

\ifllncs
\author{Pascal Schweitzer, RWTH Aachen University}
\iflipics\else
\institute{RWTH Aachen University,
  {\tt schweitzer@informatik.rwth-aachen.de} 
}
\fi
\else

\author{Pascal Schweitzer \\ RWTH Aachen University\\
{\tt schweitzer@informatik.rwth-aachen.de} 
}

\fi

\begin{document}

\maketitle

\begin{abstract}

In this paper we resolve the complexity of the isomorphism problem on
all but finitely many of the graph classes characterized by two forbidden induced subgraphs. To this end we develop
new techniques applicable for the structural and algorithmic analysis of graphs. 
First, we develop a methodology to show isomorphism completeness of the isomorphism problem on graph classes by providing a general framework unifying various reduction techniques. Second, we generalize the concept of the modular decomposition to colored graphs, allowing for non-standard decompositions. We show that, given a suitable decomposition functor, the graph isomorphism problem reduces to checking isomorphism of colored prime graphs. Third, we extend the techniques of bounded color valence and hypergraph isomorphism on hypergraphs of bounded color size as follows. We say a colored graph has generalized color valence at most~$k$ if, after removing all vertices in color classes of size at most~$k$, for each color class~$C$ every vertex has at most~$k$ neighbors in~$C$ or at most~$k$ non-neighbors in~$C$.
We show that isomorphism of graphs of bounded generalized color valence can be solved in polynomial time. 
\end{abstract}


%
%
%
%
%

\section{Introduction}

Given two graphs~$G_1$ and~$G_2$, the graph isomorphism problem asks whether there exists a bijection from the vertices of~$G_1$ to the vertices of~$G_2$ that preserves adjacency and non-adjacency.
In this paper we continue the systematic investigation of the complexity of graph isomorphism on hereditary graph classes with a focus on classes characterized by finitely many forbidden induced subgraphs as initiated in~\cite{DBLP:conf/wg/KratschS12}\ifhasappendix{}\else\  (see also~\cite{DBLP:journals/corr/abs-1208-0142})\fi. 

Given a set of finite graphs~$H_1,\ldots,H_t$ we define~$\forbind{H_1,\ldots,H_t}$ to be the class of all graphs that do not contain any~$H_i$ as an induced subgraph.  In the light of the unknown complexity status of the graph isomorphism problem, the goal in this context is typically to classify the complexity for the various graph classes into being polynomial time solvable or isomorphism complete (i.e., polynomially equivalent to graph isomorphism). Recently,  in~\cite{OtachiS13} it is shown that, assuming that graph isomorphism is not polynomial time solvable in general, there exist graph classes closed under taking (not necessarily induced) subgraphs which are of intermediate complexity. Trivially this implies the same conditional existence of hereditary graph classes (i.e., graph classes characterized by forbidden induced subgraphs) of intermediate complexity. However, the construction in~\cite{OtachiS13} intrinsically requires the use of infinitely many forbidden subgraphs. In contrast to this, in~\cite{OtachiS13} it is also shown
that there are no intermediate graph classes characterized by finitely many forbidden subgraphs. However, this statement does not carry over to forbidden induced subgraphs and the question of the existence of intermediate graph classes characterized by finitely many forbidden induced subgraphs remains open. A more precise statement of the dichotomy result in~\cite{OtachiS13} is that a graph class characterized by finitely many forbidden subgraphs has a polynomial time solvable graph isomorphism problem if one of the forbidden graphs is a union of subdivided 
stars, and is graph isomorphism complete otherwise. These graphs, the forests of subdivided stars, also play a central role in the complexity of hereditary graph classes. 

With respect to classes defined by forbidden induced subgraphs, there is a dichotomy for the isomorphism problem on~$\forbindplain{H_1}$ graphs into polynomially solvable and isomorphism complete cases. In~\cite{DBLP:conf/wg/KratschS12} the complexity of the isomorphism problem on~$\forbind{H_1,H_2}$ graphs is determined for various pairs~$(H_1,H_2)$ and the results also follow the ``polynomially solvable versus isomorphism complete'' dichotomy. The crucial cases that were not resolved are those where either~$H_1$ or~$H_2$ is a complete graph. More specifically, except for finitely many cases, the unresolved cases were shown to be polynomially equivalent to a class where one of the graphs is a complete graph.

In the light of all of these results we conjecture the following:

\begin{Conjecture}
If~$\mathcal{C} = \forbind{H_1,\ldots,H_t}$ is a graph class defined by the finite  set of forbidden induced subgraphs~$H_1,\ldots,H_t$ then graph isomorphism of graphs in~$\mathcal{C}$ is polynomial time solvable or isomorphism complete.
\end{Conjecture}

In this paper we continue the investigation of the complexity of the isomorphism problem on the graph classes focusing on the case of two forbidden subgraphs, where one of the graphs is complete. As mentioned above these are the crucial cases that were not resolved. For the resolution of those classes, techniques beyond those that were developed in~\cite{DBLP:conf/wg/KratschS12} are required. These techniques are presented in this paper leading to the following theorem:

\begin{theorem}\label{thm:main:theorem}
 On all but finitely many classes of the form~$\forbind{H_1,H_2}$ the graph isomorphism problem is polynomial time solvable or isomorphism-complete.
\end{theorem}

\emph{Contribution.}
In order to prove the theorem, in this paper three new techniques for the structural and algorithmic analysis of graphs are developed.

Firstly, we develop a methodology to show isomorphism completeness of the isomorphism problem on graph classes by providing a unifying  framework for various reductions typically used for that purpose. The advantage of this framework is that it allows a streamlined abstract way to argue why some class is isomorphism complete which boils down to a algorithmically checkable argument.

Secondly, we generalize the modular decomposition to colored graphs and define the concept of a colored modular decomposition with respect to a decomposition functor. This not only allows us to show that the graph isomorphism problem reduces to colored isomorphism of prime graphs but it also allows us to decompose graphs that are prime with respect to the classical modular decomposition. To show this reduction we also describe how to remodel an algorithm that has access to an oracle producing a complete invariant for a graph class into an algorithm that only has access to an isomorphism test of the graph class.

Thirdly, we extend the techniques of bounded color valence and hypergraph isomorphism on hypergraphs of bounded color class size as follows. We say a colored graph has generalized color valence at most~$k$ if, after removing all vertices in color classes of size at most~$k$, for each color class~$C$ every vertex has at most~$k$ neighbors in~$C$ or at most~$k$ non-neighbors in~$C$.
We show that isomorphism of graphs of bounded generalized color valence can be solved in polynomial time. This generalization allows us to perform isomorphism tests for graphs whose automorphism group cannot be forced into having bounded size composition factors even when using finitely many individualization steps. Since for such graphs alternating groups of unbounded size can appear among the composition factors, it seems that the standard group theoretic techniques cannot be directly applied. This shows that indeed new techniques were required to solve the particular cases.

We apply the three mentioned techniques to resolve the complexity of isomorphism on all but finitely many of the graph classes characterized by two forbidden induced subgraphs. For the resolved classes we either provide  a reduction from the general problem or a polynomial time algorithm. The applications of the techniques include for example showing that bipartite graphs that are free from a fixed forbidden double star refine into graphs of bounded generalized color valence. This can be used to show that isomorphism of graphs of bounded clique number with a fixed forbidden double star can be solved in polynomial time. We in turn apply this in conjunction with the colored modular decomposition technique to solve isomorphism for graphs of bounded clique number which do not contain~$P_5$ (a path of length~4) as an induced subgraph. We apply  the general reductions to show that classes of graphs without cliques of size~4 and certain unions of paths are extensive enough to have an isomorphism problem that is isomorphism complete. Furthermore we apply the modular decomposition techniques in conjunction with the bounded generalized color valence to analyze the structure of various graph classes of bounded clique number with certain forbidden forests of 
 subdivided stars.

\emph{Related work.} We refer the reader to~\cite{BabaiHandbook,KoeblerSchoeningToran:1993,SchweitzerThesis} for an introduction to the diverse complexity-theoretic results related to the isomorphism problem. There are numerous results known on the complexity of graph isomorphism of hereditary graph classes. A collection showing the problem for many classes  to be equivalent to the general problem is given by Booth and Colbourn~\cite{boothcolbourn}. In that paper the complexity of classes characterized by one forbidden subgraph~$H$ is shown to depend on whether the graph~$H$ is an induced subgraph of~$P_4$ (the path on~4 vertices).
The systematic study of classes characterized by two forbidden subgraphs was initiated in~\cite{DBLP:conf/wg/KratschS12}.

With regard to algorithms, a recent very general result, implying many results devised earlier on more special graph classes, is a theorem of Grohe and Marx~\cite{DBLP:conf/stoc/GroheM12} that shows that isomorphism of graph classes defined by a forbidden topological minor can be solved in polynomial time.

There are several applications of modules (sometimes called homogeneous sets) or some form of modular decomposition in the context of graph isomorphism. For example Goldberg's plain exponential algorithm~\cite{MR712923} uses the concept of sections, which can be seen as colored modules. Furthermore Junttila and Kaski~\cite{JunKaskiComponent} define nonuniform components within the individualization-refinement approach that also constitute colored modules. As described below, Rao~\cite{rao} also exploits the classical modular decomposition to devise an isomorphism algorithm for gem and co-gem free graphs (i.e., $\forbindplain{P_4\dunion K_1,\overline{P_4\dunion K_1}}$ graphs). His technique can be seen as a special case of the techniques using colored modular decompositions described in this paper. Other hereditary graph classes, for which graph isomorphism algorithms make use of modular decomposition, are for example subclasses of circular-arc graphs (see~\cite{DMTCS2298}).

In the early stages of the discovery of the group theoretic technique for the isomorphism problem, Luks~\cite{Luks:1982} applied it 
to show that isomorphism of graphs of bounded maximum degree can be solved in polynomial time. Babai~\cite{BabaiModerate} applied the notion of bounded color valence in his algorithm for the general isomorphism problem. Miller~(see \cite{Miller2}) applied this technique in a series of papers to perform isomorphism tests of~$k$-separable and~$k$-contactable graphs as well as isomorphism of hypergraphs of bounded color class size (see also~\cite{DBLP:conf/fsttcs/ArvindDKT10},~\cite{BabaiLuks} and~\cite{seress2003permutation}). In our generalization in this paper, only the subgraph induced by the color classes that are not of bounded size is required to exhibit bounded color valence.

Concerning dichotomies for problems on graph classes characterized by forbidden induced subgraphs there have been several studies aiming at dichotomy results for computational problem. For example for two forbidden subgraphs this has been done for the computation of the chromatic number~\cite{woegi}, dominating sets~\cite{Lozin2010} and coloring \cite{Dabrowski201434}, and list coloring \cite{DBLP:journals/dam/GolovachP14}.

Moreover there are numerous results analyzing whether the clique width of a graph class characterized by forbidden  subgraphs is bounded (see~\cite{DBLP:journals/corr/DabrowskiP14a} for an extensive list of references).

\emph{Structure of the paper.} We mainly apply the three techniques developed in this paper to classes characterized by two forbidden induced subgraphs. However, the intention behind their presentation is to allow them to be applicable in a broader sense to practical and theoretical algorithms for isomorphism of general graphs or maybe for the classes of bounded clique-width. The paper first presents the techniques and the second part, mostly contained in the appendix, explains how to apply them to various classes of two forbidden subgraphs. 

In the first part of the paper we provide preliminaries such as introducing notation and recalling basic tools (Section~\ref{sec:prelims}). We then devise a methodology to prove isomorphism completeness results (Section~\ref{sec:reductions}) and briefly describe how to simulate a complete invariant given only an isomorphism algorithm (Section~\ref{sec:invars:col:val}). After this we turn to techniques for isomorphism testing using modular decompositions (Section~\ref{sec:modular:decomp}) and devise a polynomial time algorithm for graphs of bounded generalized color valence (Section~\ref{sec:gen:col:val}).

The reduction techniques are applied to show the isomorphism completeness of various graph classes characterized by forbidden induced subgraphs (Section~\ref{sec:reductions:applied}).  
The algorithmic techniques are then applied to
graph classes with forbidden double stars (Section~\ref{sec:col:ref:and:double:stars}) and graph classes of graphs without induced paths of length~$4$ (Section \ref{sec:p5:bd:clique}). The techniques are also applied to analyze specific triangle-free graphs (Section~\ref{sec:spec:triangle}) and specific graphs of bounded clique number (Section~\ref{sec:spec:bd:clique}). We conclude by showing that together with the theorems in~\cite{DBLP:conf/wg/KratschS12} this resolves the complexity of all but finitely many graph classes defined by two forbidden induced subgraphs (Section~\ref{sec:comprehensive}). \ifhasappendix Due to space constraints, throughout the document proofs have been moved to the appendix. \fi
\ifhasappendix\marginpar{\scriptsize{\vspace{-1cm}{\color{gray}{There is also a table of contents on the last page of the appendix.}}}}\fi

\section{Preliminaries}\label{sec:prelims}

In this paper all graphs are finite, simple, undirected graphs. For a graph~$G$, by~$V(G)$ and~$E(G)$ we denote the vertex set and the edge set, respectively. By~$N_G(S)= N(S)$ we denote the neighborhood of a set~$S$, i.e., the vertices in~$V(G)\setminus S$ adjacent to some vertex in~$S$. By~$\id$ we always denote the identity map. The bipartite complement of a bipartite graph~$G$ with bipartition classes~$A$ and~$B$ is obtained by replacing~$E(G)$ with~$\{\{a,b\}\mid a\in A, b\in B\} \setminus E(G)$.

We write~$H\leq G$ if the graph~$G$ contains a graph~$H$ as an induced subgraph.
A graph~$G$ is~$\forbindONEplain{H}$ if~$H\nleq G$.
It is~\emph{$\forbindplain{H_1,\ldots,H_k}$}, if it is~$\forbindONEplain{H_i}$ for all~$i$. A graph class~$\classC$ is~\emph{$\forbindONEplain{H}$} (respectively~\emph{$\forbind{H_1,\ldots,H_k}$}) if this is true for all~$G\in\classC$. A graph class~$\classC$ is \emph{hereditary} if it is closed under taking induced subgraphs. The class~$\forbind{H_1,\ldots,H_k}$ is the class of all~$\forbindplain{H_1,\ldots,H_k}$ graphs. Note that each class~$\forbind{H_1,\ldots,H_k}$ is hereditary.
We say a graph~$G$ \emph{contains} a graph~$H$ (as an induced subgraph) if an induced subgraph of~$G$ is isomorphic to~$H$.

By~$I_t$,~$K_t$,~$P_t$, and~$C_t$ we denote the independent set, the clique, the path, and the cycle on~$t$ vertices, respectively. The \emph{clique number} of a graph~$G$ is the largest integer~$t$ such that~$G$ contains~$K_t$. By~$H\dunion H'$ we denote the disjoint union of~$H$ and~$H'$; we use~$tH$ for the disjoint union of~$t$ copies of the graph~$H$. By~$\overline{G}$ we denote the (edge) complement of~$G$. The graph~$\overline{K_2\dunion I_2}$, i.e., the graph obtained from~$K_4$ by deleting an edge, is called the \emph{diamond}.

\hereorinappendixstatement{
\begin{figure}[t]
\centering
\subfloat[width=10in][The graph~$H(1,b,0)$ that plays a central role in Section~\ref{sec:spec:bd:clique}.]
{\makebox[.32\textwidth]{\begin{tikzpicture}[scale=0.7,thick]
\draw (0,-1) -- (0,0) -- (1.5,1) -- (1,0);
\draw (1.5,1) -- (3,0);
\drawvertex{(0,-1)}
\drawvertex{(0,0)}
\drawvertex{(1,0)}
\drawvertex{(3,0)}
\drawvertex{(1.5,1)}
\draw (0,-2.5) node {\mbox{}};
\draw (2,0) node {$\ldots$};
\draw (2,-0.7) node {$\underbrace{\hskip1.6cm}_{\,}$};
\draw (2,-1.2) node {\smash{$b$}};
\end{tikzpicture}}

\label{fig:subvidided:star:H1b0}
}
\hspace{1cm}
\subfloat[][The graph~$H(1,0,b,1)$ that plays a central role in Section~\ref{sec:spec:triangle}.]{\makebox[.32\textwidth]{
\begin{tikzpicture}[scale=0.7,thick]
\draw (0,-1) -- (0,0) -- (1.5,1) -- (1,0);
\draw (1.5,1) -- (3,0);
\draw (0,-1) -- (0,-2);
\drawvertex{(0,-2)}
\drawvertex{(0,-1)}
\drawvertex{(0,0)}
\drawvertex{(1,0)}
\drawvertex{(3,0)}
\drawvertex{(4,0)}
\drawvertex{(1.5,1)}
\draw (2,0) node {$\ldots$};
\draw (0,-2.5) node {\mbox{}};
\draw (2,-0.7) node {$\underbrace{\hskip1.6cm}_{\,}$};
\draw (2,-1.2) node {\smash{$b$}};
\end{tikzpicture}}

\label{fig:subvidided:star:H10b1}
}
\caption{Examples of subdivided stars with possibly added isolated vertices.}\label{fig:taxonomy:of:H:graphs}
\end{figure}
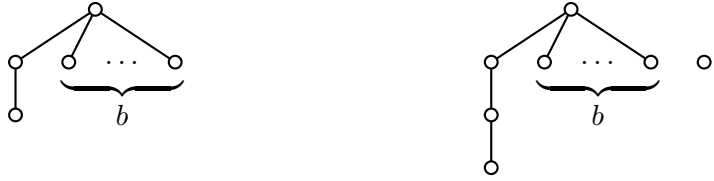
}
A \emph{star} is a graph isomorphic to the complete bipartite graph~$K_{1, t}$ for some positive integer~$t$.
A \emph{subdivided star} is a possibly repeated subdivision of a star. 
\ifhasappendix \else Since any subdivided star is a tree, we call the disjoint union of subdivided stars a \emph{forest of subdivided stars}.\fi
If a subdivided star has a vertex of degree at least 3 then this vertex is unique and called
 the \emph{center}. For non-negative integers~$a_0,\ldots,a_t$ with~$a_t>0$, we define the graph~$H(a_t,\ldots,a_1,a_0)$ to be the disjoint union of an independent set of size~$a_0$ with the following subdivided star $H$. The star~$H$ is the subdivided star that for~$i\in\{1,\ldots,t\}$ has exactly~$a_i$ leaves at distance~$i$ from the center and no other leaves. (If~$\sum_{i=1}^t(a_i) < 3$ the center of~$H$ is defined so that the graph~$H$ is a path whose two leaves have suitable distances from that center.) \ifhasappendix
 {{\marginpar{\scriptsize{\vspace{-1cm}{\color{gray}{See Figure~\ref{fig:taxonomy:of:H:graphs} in the Appendix for examples of such graphs.}}}}}} \else
  Examples of such graphs are depicted in Figure~\ref{fig:taxonomy:of:H:graphs}. \fi

In~\cite{DBLP:conf/wg/KratschS12} it is shown that isomorphism of~$\forbindplain{H_1,\ldots,H_t}$ graphs is isomorphism complete unless one of the forbidden graphs is a forest of subdivided stars.

In this paper a \emph{colored graph} is a vertex colored graph whose coloring does not need to be proper (i.e., adjacent vertices can have the same color). Isomorphisms between colored graphs are required to respect the colors\ifhasappendix . \else, that is, they must map vertices to vertices of the same color. \fi A \emph{singleton} is a vertex with a unique color.
The \emph{naive vertex refinement} algorithm, or 1-dimensional Weisfeiler-Lehman algorithm, is a standard technique of repeatedly recoloring the vertices, refining the partition induced by the colors by using the multiplicity of colors appearing among the neighbors of a vertex (see for example~\cite{SchweitzerThesis}).
It has the property that after the refinement, the number of neighbors a vertex~$v$ has in a color class~$C$ only depends on the colors of~$v$ and~$C$.
A graph has \emph{color valence at most~$k$} if for every vertex~$v$ and every color class~$C$ there are at most~$k$ neighbors of~$v$ in~$C$ or there are at most~$k$ non-neighbors of~$v$ in~$C$.

\section{Reductions}\label{sec:reductions}
\differingappendixstatement{}{\section{Proofs of Section~\ref{sec:reductions}}}

In this section we develop a systematic approach to proving isomorphism invariant reductions. This provides general means to construct isomorphism-complete graph classes. 
Standard reductions like subdividing, taking complements, and adding isolated vertices fall into this framework. Likewise, most reductions performed in~\cite{DBLP:conf/wg/KratschS12} and also various reductions in~\cite{boothcolbourn} fall into this framework.

\begin{definition}
Let~$J$ be a finite set and~$L \colon J \times J \rightarrow \{A,N\}$ be a labeling assigning every ordered pair of vertices the label~$A$ for adjacent or~$N$ for non-adjacent. Moreover let~$L_N \colon J \times J \rightarrow \mathbb{N} \cup\{\infty\} $ be a labeling assigning every ordered pair of vertices an integer or infinity.

A graph~$G$ belongs to \emph{the class encoded by the labeled graph~$(J,L,L_N)$} if there exists a map~$\phi\colon V(G) \rightarrow J$ such that the following hold:
\begin{enumerate}
\item If~$v \in V(G)$ and~$j \in J$ such that~$L(\phi(v),j) = A$ and~$L_N(\phi(v),j) \neq \infty$ then there are at most~$L_N(\phi(v),j)$ vertices~$v'\in V(G) \setminus \{v\}$ that are non-adjacent to~$v$ such that~$\phi(v') = j$.

\item If~$v \in V(G)$ and~$j \in J$ such that~$L(\phi(v),j) = N$ and~$L_N(\phi(v),j) \neq \infty$ then there are at most~$L_N(\phi(v),j)$ vertices~$v'\in V(G) \setminus \{v\}$ that are adjacent to~$v$ such that~$\phi(v') = j$.
\end{enumerate}
\end{definition}

In this definition, the triple~$(J,L,L_N)$ should be thought of as a generalized graph. The graph class encoded by~$(J,L,L_N)$ then contains graphs that can be obtained by replacing the elements of~$J$ with sets of vertices. In some contexts,  
this is referred to as blowing-up the elements of~$J$. Adjacency of the new vertices is essentially governed by the adjacency of the original graph. However,  the values of~$L_N$ control the number of exceptions to this rule that are allowed per vertex.

The definition captures various constructions that are used in graph theory. A first class of examples form the complete multipartite graphs, which turn out to be graphs modeled with the function~$L$ satisfying~$L(x,y) = N$ if and only if~$x= y$ and the function~$L_N$ being the constant function evaluating to~0.

A second, more general class of examples modeled by the definition is the graphs of bounded color valence, which are graphs that frequently appear in the context of graph isomorphism.  
Such graphs are obtained whenever~$L_N$ is a bounded function. The coloring corresponds to the map~$\phi$.

A third example is the class of graphs that map homomorphically onto a finite graph~$H$. This class is obtained by letting~$(J,L)$ model the graph~$H$ (i.e.,~$V(H) = J$ and~$L(x,y) = A$ if and only if~$\{x,y\}$ is an edge of~$H$) and setting~$L_N (x,y) = 0$ if~$L(x,y) = N$ and~$L_N (x,y) = \infty$ otherwise.

Of interest to us is the complexity of the isomorphism problem of graph classes encoded by a triple~$(J,L,L_N)$. It turns out that when~$L_N$ is a bounded function then isomorphism can be solved in polynomial time.

\begin{theorem}
Let~$(J,L,L_N)$ be a triple that encodes a graph class. If all values of~$L_N$ are finite then isomorphism of graphs encoded by~$(J,L,L_N)$
can be solved in polynomial time. 
\end{theorem}
\proofatend
We show the statement by induction on~$|J|$. However, we show a statement that is slightly stronger. Namely, we show that there exists an algorithm that solves the isomorphism problem for input pairs where only one of the graphs is required to be encoded by~$(J,L,L_N)$. Moreover we require that the algorithm  does not produce any incorrect answers when queried with a pair of graphs both of which are not encoded by~$(J,L,L_N)$. That is, when both graphs are not encoded by~$(J,L,L_N)$ the algorithm either correctly determines whether the graphs are isomorphic or states that the answer was not determined.

For~$|J|=1$ the encoded graphs have bounded degree or bounded co-degree, so isomorphism can be solved in polynomial time~\cite{Luks:1982}. Let~$(J,L,L_N)$ be an encoding with~$|J|>1$.  We first describe a reduction rule. Suppose there are~$j_1,j_2\in J$ such that for all~$j\in J$ we have~$L(j_1,j)= L(j_2,j)$ and~$L(j,j_1)= L(j,j_2)$. We now argue that in this case we can replace~$J$ by a smaller set at the expense of increasing the values in~$L_N$. To this end, let~$\pi$ be the projection that maps~$j_2$ to~$j_1$, i.e., the map that satisfies~$\pi(j) = j$ if~$j\neq j_2$ and~$\pi(j_2)= j_1$. We claim that every graph encoded by~$(J,L,L_N)$ is also encoded  by~$(J',L',L'_N)$ where~$J' = J\setminus\{j_2\}$,~$L|_{J\setminus\{j_2\}}$ is the restriction of~$L$ to~$J'$, and~$L'(j,j')= 2 \cdot \max\{ L(\ell,\ell')\mid \ell\in \pi^{-1}(j),\ell' \in \pi^{-1}(j')\}$. Indeed, if for a graph~$G$ a map~$\phi \colon V(G) \rightarrow J$ is an encoding by~$(J,L,L_N)$ then the map~$\pi \circ \phi$ is an encoding by~$(J',L',L'_N)$.
By induction we can solve isomorphism for graphs encoded by~$(J',L',L'_N)$ in polynomial time.

Suppose now that the above situation does not arise. We can assume that for every map that encodes one of the input graphs by~$(J,L,L_N)$, the map to~$J$ is surjective. Otherwise we can solve the isomorphism problem on a smaller encoding by induction. 

We will now guess an encoding~$\phi_1$ of the input graph~$G_1$.
For each~$j$ in~$J$ we guess a vertex~$v\in V(G_1)$ with~$\phi_1(v) = j$. Let~$M$ be the set of guessed vertices. We now guess all vertices that form an exception with a vertex in~$M$. That is, we guess all vertices~$u$ for which there exists a~$v\in M$ such that~$L(\phi_1(v),\phi_1(u)) = A$ but~$u$ and~$v$ are not adjacent or such that~$L(\phi_1(v),\phi_1(u))= N$ but~$u$ and~$v$ are adjacent. Since~$L_N$ is bounded, the possible number of vertices that form an exception with a vertex in~$M$ is bounded. For all these vertices we guess their image under~$\phi_1$. By the reduction described above, for all non-exceptional vertices not in~$M$ there is only one choice for their image under~$\phi_1$ not violating the encoding or making them an exception. Thus~$\phi_1$ is completely determined by the guesses we have performed. The process of guessing an image under~$\phi_1$ is performed only for a bounded number of vertices. Thus, there are only polynomially many options.

Similarly for~$G_2$ there are only polynomially many options for suitable encodings. We can also check validity of an encoding in polynomial time. Once an encoding is determined, the graphs have bounded color valence and isomorphism can be checked in polynomial time.
\endproofatend

The theorem implicitly shows that isomorphism of graphs that have bounded color valence for some coloring that uses a bounded number of color classes can be solved in polynomial time, even if the color classes are not given.

\begin{corollary}
For every positive integer~$c$, graph isomorphism of graphs whose vertices can be partitioned into~$c$ color classes such that the graph has color-valence at most~$c$ can be solved in polynomial time.
\end{corollary}
\proofatend
Every graph that can be partitioned into~$c$ color classes such that the graph has color-valence at most~$c$ can be encoded by a triple~$(J,L,L_N)$ with~$|J| \leq c$ and a function~$L_N$ bounded by~$c$. There are only finitely many encodings of this type. For each of these we can apply the algorithm from the proof of the theorem. The input graphs are isomorphic if and only if at least one of the calls to the algorithm claims them to be isomorphic.
\endproofatend

While encodings can be used to show polynomial time solvability of certain graph classes, they can also be used to show hardness results as follows. 

\begin{definition}
An encoding~$(J,L,L_N)$ is a \emph{simple path encoding} in case~$L$ is symmetric (i.e., if~$L(j,j')= L(j',j)$ holds) and there is a sequence of vertices~$(p_1,\ldots,p_t)$ of length at least~2 in~$J$ such that~$L_N(p_1,p_2) = \infty$,~$L_N(p_t,p_{t-1}) \geq 2$ and for all~$k$ we have~$L_N(p_k,p_{k+1}) \geq 1$ and~$L_N(p_{k+1},p_{k}) \geq 1$.
\end{definition}

Intuitively, a simple path encoding allows enough freedom to encode bipartite graphs with one bipartition having vertices of degree two. We can formally prove this statement in the form of a reduction.

\begin{theorem}\label{thm:simple:path:enc}

The class of graphs encoded by a simple path encoding is graph isomorphism complete.
\end{theorem}
\proofatend
First we describe a reduction and then show it is isomorphism invariant. 
Let~$(J,L,L_N)$ be a simple path encoding with the path~$P = (p_1,\ldots,p_t)$.
Let~$G = (V,E)$ be a graph. We now describe a reduction that transforms~$G$ into a graph~$G'$ that can be encoded by~$(J,L,L_N)$.
The vertex set of~$G'$ is~$V(G') = V(G) \cup E(G) \cup \{(v,\{v,v'\},k)\mid \{v,v'\} \in E(G) \text{ and } k\in  \{2,\ldots, t-1\} \}$. We define a map~$\phi\colon V(G') \rightarrow J$ such that for all~$v\in V(G)$ we have~$\phi(v)= p_1$ and for all~$e\in E(G)$ and~$i\in \{2,\ldots, t-1\}$ we have~$\phi(e) = p_t$ and~$\phi(\{e,i\}) = p_i$. 
The edge set~$E(G')$ will be formed as a symmetric difference as follows:
\begin{itemize}
\item Let $E_1$ be the set of pairs~$\{u,u'\} \subset V(G')$ for which~$L(\phi(u),\phi(u')) = A$.
\item Let~$E_2$ be the set of pairs~$\{(v,\{v,v'\},i),(v,\{v,v'\},i+1)\mid \{v,v'\}\in E(G) \text{, } i \in \{2,\ldots,t-2\}\}$.

\item In the case~$t>2$ we define $E_3$ to be the set of pairs~$\{v,(v,\{v,v'\},2)\mid \{v,v'\}\in E(G)\}$ and~$E_4$ to be the set of pairs~$\{\{v,v'\},(v,\{v,v'\},t-1)\mid \{v,v'\}\in E(G)\}$.
\item If~$t = 2$ we set~$E_3$ as~$\{v,\{v,v'\}\}$ and~$E_4$ is the empty set.
\end{itemize}
The edge set~$E(G')$ is~$E_1 \bigtriangleup (E_2 \cup E_3 \cup E_4)$, where~$\bigtriangleup$ denotes the symmetric difference. 

By following the construction, given a graph~$G$, the graph~$G'$ can be constructed in polynomial time. The map~$\phi$ is an encoding and shows that the graph~$G$ can be encoded by the triple~$(J,L,L_N)$. Since all steps of the construction are isomorphism invariant, given two isomorphic graphs~$G_1$ and~$G_2$ the respective graphs produced by the construction to each of them are isomorphic as well. 

To finish the proof we need to show the following: there is a graph isomorphism complete class of graphs~$\mathcal{C}$ such that for any two non-isomorphic graphs in~$\mathcal{C}$ the construction produces two non-isomorphic graphs.

For this let~$\mathcal{C}$ be the class of graphs that have the following properties: the number of edges is at least~$5n$. The minimum degree and the maximum co-degree are at least~$4$. This class can be seen to be isomorphism complete by using the operation that adds 6 non-adjacent vertices to a graph and making them adjacent to every other vertex. There are various other reductions showing that this is an isomorphism complete class, for example the class of self-complementary graphs fulfills these properties, which is an isomorphism complete class~\cite{Colbourn:Colbourn:self:compl}.

To show that the construction described above maps non-isomorphic graphs from~$\mathcal{C}$ to non-isomorphic graphs it suffices to reconstruct the labeling. To show that the labeling is reconstructible, we use a degree argument. The degree of a vertex with label~$p_1$ is~$d + k_1 n - \varepsilon + k_2 m $ where~$d$ is the degree of the vertex in~$G$, the variables~$k_1, \varepsilon \in \{0,1\}$ depend on the value of~$L(p_1,p_1)$,  and~$k_2\in \{0,\ldots, |J|\}$ is a constant that depends on the values of~$L(p_1,p_i)$.

The degree of a vertex with a label~$p_i$ different from~$p_1$ is~$k_1 n + k_2 m + k_3$, where again~$k_1$,~$k_2$ and~$k_3$ depend on the labeling~$L$. More precisely,~$k_1\in \{0,1\}$ depends on~$L(p_i,p_1)$,~$k_2$ depends on the values of~$L(p_i,p_{i'})$ and~$k_3$ is an integer with an absolute value that is bounded above by~$2$ arising from exceptions allowed by~$L_N$.

Since~$m \geq 5n$ the degree of vertices with label~$p_1$ cannot be expressed as a sum with constraints of the latter kind (i.e., as~$k_1 n + k_2 m + k_3$). Likewise the degree of vertices with a label different from~$p_1$ cannot be expressed with constraints of the former kind (i.e., as~$d + k_1 n - \varepsilon + k_2 m $). This implies that the set of vertices with label~$p_1$ is reconstructible. Since for~$i>1$ the set of vertices with label~$p_i$ contains exactly the vertices that have a neighbor and a non-neighbor among the vertices with label~$p_{i-1}$, 
we can reconstruct~$\phi$ and thus the partition. Once the partition is reconstructed, the edge sets~$E_i$ can be reconstructed, and thus the original graph can be reconstructed. This shows that the reduction preserves isomorphism for graphs in~$\mathcal{C}$ and finishes the proof.
\endproofatend

\ifhasappendix
\else
In the next section, we give examples of reductions that apply the theorem. We require these results later to resolve the complexity of various hereditary graph classes of bounded clique number.

\fi
\ifhasappendix
In the appendix in Section~\ref{sec:reductions:applied}, we give examples of reductions that apply the theorem. We require these results to resolve the complexity of various hereditary graph classes of bounded clique number.\fi
\hereorinappendixstatement{
\section{Isomorphism-complete classes}\label{sec:reductions:applied}

The reduction techniques of \ifhasappendix{Section~\ref{sec:reductions} }\else{the previous section}\fi provide a systematic approach to show isomorphism completeness results for graphs characterized by finitely many forbidden induced subgraphs. The following examples show how to apply the techniques in a mechanical way to show isomorphism completeness of graph classes characterized by forbidden induced subgraphs.

\begin{theorem}\label{thm:K4:and:various:free:iso:compl}
The graph isomorphism problem on the class of~$\forbindplain{2K_2 \dunion K_1,K_4}$ graphs, of~$\forbindplain{P_6,P_4\dunion P_2, K_4}$ graphs, and of~$\forbindplain{\HabcGraphFOUR{1}{0}{3}{0}, K_4}$ graphs is graph isomorphism complete.
\end{theorem}
\begin{proof}
We first show the statement for the class of~$\forbindplain{2K_2 \dunion K_1,K_4}$ graphs. For this, let~$(\{p_1,p_2,p_3\}, L, L_N)$ be the encoding given by setting~$L(v,v')= N$ if and only if~$\{v,v'\}\neq \{p_1,p_3\}$ and~\[L_N(v,v') =    \begin{cases} 
\infty  &\text{ if } (v,v')= (p_1,p_2) \\
1  &\text{ if } (v,v')= (p_2,p_1) \\
1  &\text{ if } (v,v')= (p_2,p_3) \\
2  &\text{ if } (v,v')= (p_3,p_2) \\
0  &\text{ otherwise.}\end{cases}\]
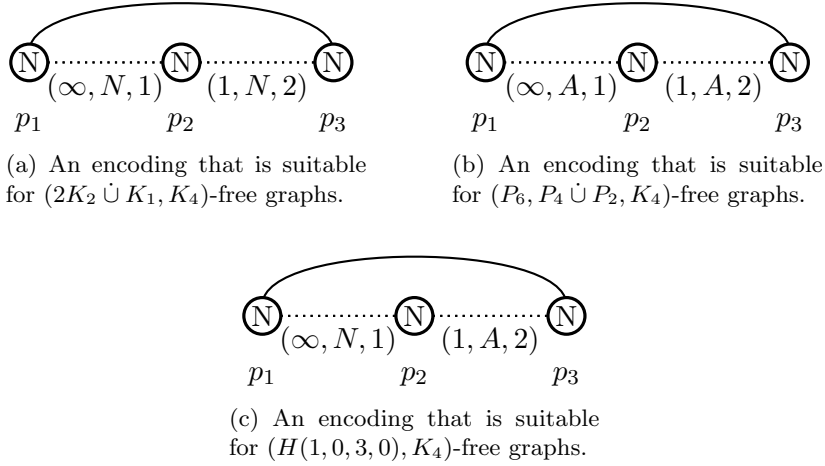
\begin{figure}[t]
\centering
\subfloat[][An encoding that is suitable for $\forbindplain{2K_2 \dunion K_1,K_4}$ graphs.]
{\begin{tikzpicture}[scale=1,thick]
\node (v1) at (0,0) {};
\node (v2) at (2,0) {};
\node (v3) at (4,0) {};

\draw (v1) .. controls (0,1) and (4,1) .. (v3);
\draw[dotted] (v1) -- (v2) -- (v3);
\path[dotted] (v1) edge node[midway,below] {$(\infty,N,1)$} (v2);
\path[dotted] (v2) edge node[midway,below] {$(1,N,2)$} (v3);

\node[vertex] at (v1)  {N};
\node[vertex] at (v2)  {N};
\node[vertex] at (v3)  {N};
\draw (0,-0.8) node {$p_1$};
\draw (2,-0.8) node {$p_2$};
\draw (4,-0.8) node {$p_3$};
\end{tikzpicture}

\label{fig:encoding:1}
}
\hspace{1cm}
\subfloat[][An encoding that is suitable for~$\forbindplain{P_6,P_4\dunion P_2, K_4}$ graphs.]{
\begin{tikzpicture}[scale=1,thick]
\node (v1) at (0,0) {};
\node (v2) at (2,0) {};
\node (v3) at (4,0) {};

\draw (v1) .. controls (0,1) and (4,1) .. (v3);
\draw[dotted] (v1) -- (v2) -- (v3);
\path[dotted] (v1) edge node[midway,below] {$(\infty,A,1)$} (v2);
\path[dotted] (v2) edge node[midway,below] {$(1,A,2)$} (v3);

\node[vertex] at (v1)  {N};
\node[vertex] at (v2)  {N};
\node[vertex] at (v3)  {N};
\draw (0,-0.8) node {$p_1$};
\draw (2,-0.8) node {$p_2$};
\draw (4,-0.8) node {$p_3$};
\end{tikzpicture}

\label{fig:encoding:2}
}

\subfloat[][An encoding that is suitable for~$\forbindplain{\HabcGraphFOUR{1}{0}{3}{0}, K_4}$ graphs.]{
\begin{tikzpicture}[scale=1,thick]
\node (v1) at (0,0) {};
\node (v2) at (2,0) {};
\node (v3) at (4,0) {};

\draw (v1) .. controls (0,1) and (4,1) .. (v3);
\draw[dotted] (v1) -- (v2) -- (v3);
\path[dotted] (v1) edge node[midway,below] {$(\infty,N,1)$} (v2);
\path[dotted] (v2) edge node[midway,below] {$(1,A,2)$} (v3);

\node[vertex] at (v1)  {N};
\node[vertex] at (v2)  {N};
\node[vertex] at (v3)  {N};
\draw (0,-0.8) node {$p_1$};
\draw (2,-0.8) node {$p_2$};
\draw (4,-0.8) node {$p_3$};
\end{tikzpicture}

\label{fig:encoding:3}
}
\caption{The figures show two encodings~$(J,L,L_N)$ with symmetric function~$L$.}
\end{figure}%
This encoding is depicted in Figure~\ref{fig:encoding:1}.
By Theorem~\ref{thm:simple:path:enc} it suffices to show that graphs encoded by~$(\{p_1,p_2,p_3\}, L, L_N)$ are~$\forbindplain{2K_2 \dunion K_1}$ and~$\forbindplain{K_4}$. This is equivalent to showing that neither~$2K_2 \dunion K_1$ nor~$K_4$ can be encoded. For~$K_4$ this is obvious since all encoded graphs are tripartite. For~$2K_2 \dunion K_1$ we conduct a case distinction. First note that, under a tentative encoding, it cannot be the case that all elements of~$\{p_1,p_2,p_3\}$ are in the image of the encoding, as otherwise two non-adjacent vertices would be mapped to~$p_1$ and~$p_3$ each.
If the isolated vertex is mapped to~$p_1$ (respectively~$p_3$) then all other vertices must be mapped to~$p_2$ or~$p_3$ (respectively~$p_2$ or~$p_3$). Since end points of edges cannot be mapped to the same element of~$\{p_1,p_2,p_3\}$, two vertices are mapped to~$p_2$. This implies that some vertex mapped to~$p_2$ has two non-neighbors mapped to~$p_1$ (respectively~$p_3$) contradicting that~$L_N(p_2,p_1) = 1$ (respectively ~$L_N(p_2,p_3) = 1$).
In the case where the isolated vertex is mapped to~$p_2$ there are two vertices mapped to~$p_1$ or two vertices mapped to~$p_3$ which similarly leads to a contradiction.

To show the second part of the theorem for~$\forbindplain{P_6,P_4\dunion P_2, K_4}$ graphs we use the encoding~$(\{p_1,p_2,p_3\}, L', L_N)$ with~$L'(v,v')= N$ if and only if~$v=v'$ (and the same function~$L_N$ as in the first part of the proof). This encoding is depicted in Figure~\ref{fig:encoding:2}. Analogously to the first part, it suffices for us to show that the graphs~$P_6$,~$P_4\cup P_2$ and~$K_4$ cannot be encoded.
Again, the encoded graphs are tripartite so~$K_4$ cannot be encoded. 

We first note that for any encoding of~$P_4$, some vertex is mapped to~$p_1$ and some vertex is mapped to~$p_3$. Indeed, otherwise, some vertex is mapped to~$p_2$, two neighbors are mapped to~$p_1$, or two neighbors mapped to~$p_3$.

Suppose there is an encoding of~$P_6$. Since~$P_4 \leq P_6$, some vertex is mapped to~$p_1$ and some vertex is mapped to~$p_3$. Since every non-empty complete bipartite graph on 4 vertices contains a vertex of degree 3 or a cycle, this implies that in any encoding at most 3 vertices can be mapped to~$p_1$ or~$p_3$. Thus 3 vertices must be mapped to~$p_2$. The path~$P_6$ only has two independent sets of size~$3$. The vertices not mapped to~$p_2$ also form an independent set, so they must be all mapped to~$p_1$ or all mapped to~$p_3$ which contradicts the observation about encodings of~$P_4$.

Suppose now that there is an encoding of~$P_4\dunion P_2$. 
We already know that some vertex of~$P_4$ is mapped to~$p_1$ and some vertex is mapped to~$p_3$. However, some vertex of~$P_2$ must also be mapped to~$p_1$ or~$p_3$ and this vertex would thus have to be adjacent to a vertex of~$P_4$.

To show the third part of the theorem for~$\forbindplain{\HabcGraphFOUR{1}{0}{3}{0}, K_4}$ graphs we use the encoding~$(\{p_1,p_2,p_3\}, L', L_N)$ with~$L'(v,v')= A$ if and only if~$\{v,v'\}=\{p_2,p_3\}$ or~$\{v,v'\}=\{p_1,p_3\}$ (and the same function~$L_N$ as above). This encoding is depicted in Figure~\ref{fig:encoding:3}. By the same arguments as before~$K_4$ cannot be encoded, so it suffices to show that~$\HabcGraphFOUR{1}{0}{3}{0}$ cannot be encoded. Let~$h$ be the vertex of degree~$4$ in the graph~$\HabcGraphFOUR{1}{0}{3}{0}$ and let~$x_1$,~$x_2$ and~$x_3$ be the vertices on the path from~$h$ to the leaf of distance~$3$ from the center. Suppose~$h$ is mapped to~$p_1$. In this case~$x_1$ could not be mapped to~$p_2$ since~$x_2$ would have to be mapped to~$p_3$, which contradicts that~$x_2$ is not adjacent to~$h$. Thus~$x_1$ would have to be mapped to~$p_3$. Then~$x_2$ is mapped to~$p_1$, since otherwise~$x_3$ cannot be mapped anywhere. This implies that all leaves are mapped to~$p_2$. However, then~$x_1$, which is mapped to~$p_3$, has~$3$ non-neighbors which are mapped to~$p_2$ contradicting the requirements of the encoding.
Supposing next that~$h$ is mapped to~$p_2$, then at least one, and thus all of its neighbors are mapped to~$p_3$. Since~$x_2$ must then be mapped to~$p_2$ this would mean that~$x_2$ has~$2$ non-neighbors that are mapped to~$p_3$ contradicting the encoding.
Finally suppose~$h$ is mapped to~$p_3$. This implies that one of the vertices~$x_2$ and~$x_3$ is mapped to~$p_2$ while the other is mapped to~$p_3$. This in turn implies that all leaf vertices adjacent to~$h$ are mapped to~$p_2$ which shows that some vertex in~$\{x_2,x_3\}$ is mapped to~$p_3$ but has 3 non-neighbors that are mapped to~$p_2$ which yields a contradiction.
\end{proof}

As can been seen in the proof of the previous theorem, given an encoding, checking whether no~$\forbindONEplain{H}$ graph can be encoded amounts to checking whether~$H$ itself can be encoded, which is a finite computation that can in fact be fully automated. We consider another example concerned with bipartite graphs.

\begin{theorem}\label{thm:bipart:2P3:K1:free:iso:compl}
The isomorphism problem for bipartite~$\forbindplain{2P_3 \dunion K_1}$ graphs is graph isomorphism complete.
\end{theorem}
\begin{proof} 
Let~$(\{p_1,p_2,p_3,p_4\}, L, L_N)$ be the encoding depicted in Figure~\ref{fig:encoding:cross:cocycle} with~$L(v,v')= N$ if and only if~$\{v,v'\}\neq \{p_1,p_4\}$ and~\[L_N(v,v') =    \begin{cases} 
\infty  &\text{ if } (v,v')= (p_1,p_2) \\
1  &\text{ if } (v,v')= (p_2,p_1) \\
1  &\text{ if } (v,v')= (p_2,p_3) \\
1  &\text{ if } (v,v')= (p_3,p_2) \\
1  &\text{ if } (v,v')= (p_3,p_4) \\
2  &\text{ if } (v,v')= (p_4,p_3) \\
0  &\text{ otherwise.}\end{cases}\]

For every encoding of a graph the pre-images of~$\{p_1,p_3\}$ and~$\{p_2,p_4\}$ must be independent sets and thus the encoded graphs are bipartite. Let~$\pi$ be an encoding of~$2P_3 \dunion K_1$ and consider the bipartition induced by the pre-images just considered.
We first observe that no independent set of~$2P_3 \dunion K_1$ can be partitioned such that there are 2 vertices in each bipartition class. Indeed, if this were the case, then two vertices are mapped to~$p_2$ or one pair of endpoints is mapped to~$p_3$. Due to the requirements of the encoding 
this implies that the other pair is mapped to~$p_1$ or~$p_4$, which again gives a contradiction to the requirements of the encoding.

Since~$P_3$ is connected, the endpoints of each~$P_3$ are contained in the same bipartition class. 
Thus by our observation about independent sets of size 4, all end points of the paths are mapped to the same bipartition class. Furthermore, by the same reason, the isolated vertex must be mapped to the same bipartition class. This means the other two vertices have 3 non-neighbors in the other bipartition class. Thus they are mapped to~$p_1$. This implies that all other vertices are mapped to~$p_2$ giving a contradiction to~$L(p_2,p_1)$ for the isolated vertex.
\end{proof}

\begin{figure}[t]
\centering
\begin{tikzpicture}[scale=1,thick]
\node (v1) at (0,0) {};
\node (v2) at (2,0) {};
\node (v3) at (4,0) {};
\node (v4) at (6,0) {};

\draw (v1) .. controls (0,1) and (6,1) .. (v4);
\draw[dotted] (v1) -- (v2) -- (v3) -- (v4);
\path[dotted] (v1) edge node[midway,below] {$(\infty,N,1)$} (v2);
\path[dotted] (v2) edge node[midway,below] {$(1,N,1)$} (v3);
\path[dotted] (v3) edge node[midway,below] {$(1,N,2)$} (v4);

\node[vertex] at (v1)  {N};
\node[vertex] at (v2)  {N};
\node[vertex] at (v3)  {N};
\node[vertex] at (v4)  {N};
\draw (0,-0.8) node {$p_1$};
\draw (2,-0.8) node {$p_2$};
\draw (4,-0.8) node {$p_3$};
\draw (6,-0.8) node {$p_4$};
\end{tikzpicture}
\caption{An encoding that is suitable for bipartite $\forbindplain{2P_3 \dunion K_1}$ graphs.}
\label{fig:encoding:cross:cocycle}
\end{figure}%

On a more intuitive level, one can also interpret the proof as an argument showing that~$2P_3\dunion K_1$ cannot be partitioned into two parts such that its bipartite complement is a forest of subdivided stars.

Finally, we conclude with one more example-application of the reduction technique. 
\begin{theorem}\label{thm:H1020:K5:free:iso:compl}
The graph isomorphism problem on~$\forbindplain{\HabcGraphFOUR{1}{0}{2}{0},K_5}$ graphs is graph isomorphism complete.
\end{theorem}

\begin{figure}[t]
\centering
\begin{tikzpicture}[scale=1,thick]
\node (v1) at (0,0) {};
\node (v2) at (2,0) {};
\node (v3) at (4,0) {};
\node (v4) at (6,0) {};

\draw (v1) .. controls (0,1.3) and (6,1.3) .. (v4);
\draw (v1) .. controls (0,0.8) and (4,0.8) .. (v3);
\draw (v2) .. controls (2,0.8) and (6,0.8) .. (v4);

\draw[dotted] (v1) -- (v2) -- (v3) -- (v4);
\path[dotted] (v1) edge node[midway,below] {$(\infty,N,1)$} (v2);
\path[dotted] (v2) edge node[midway,below] {$(1,A,1)$} (v3);
\path[dotted] (v3) edge node[midway,below] {$(1,N,2)$} (v4);

\node[vertex] at (v1)  {N};
\node[vertex] at (v2)  {N};
\node[vertex] at (v3)  {N};
\node[vertex] at (v4)  {N};
\draw (0,-0.8) node {$p_1$};
\draw (2,-0.8) node {$p_2$};
\draw (4,-0.8) node {$p_3$};
\draw (6,-0.8) node {$p_4$};
\end{tikzpicture}
\caption{An encoding that is suitable for bipartite $\forbindplain{2P_3 \dunion K_1}$ graphs.}
\label{fig:encoding:H1020:K5}
\end{figure}%

\begin{proof}
Let~$(\{p_1,p_2,p_3,p_4\}, L, L_N)$ be the encoding depicted in Figure~\ref{fig:encoding:H1020:K5} with~$L(v,v')= N$ if and only if~$\{v,v'\}= \{p_1,p_2\}$,~$\{v,v'\}= \{p_3,p_4\}$   or~$v= v'$, and with~$L_N$ defined as in the proof of Theorem~\ref{thm:bipart:2P3:K1:free:iso:compl}. Since graphs encoded by this encoding are~$4$-partite, it suffices to show that~$\HabcGraphFOUR{1}{0}{2}{0}$ cannot be encoded. Let~$h$ be the vertex of degree 3 in~$\HabcGraphFOUR{1}{0}{2}{0}$ and let~$x_1$,~$x_2$ and~$x_3$ be the vertices on the path of length~$4$ starting in~$h$. In any encoding, the vertices~$x_3$ and~$x_2$ are mapped to different elements of~$\{p_1,p_2,p_3,p_4\}$. Since there is an edge in~$\HabcGraphFOUR{1}{0}{2}{0}$ whose endpoints are non-adjacent to the endpoints of the edge~$\{x_2,x_3\}$ the possible images of~$\{x_2,x_3\}$ under a tentative encoding are~$\{p_1,p_2\}$,~$\{p_2,p_3\}$ and~$\{p_3,p_4\}$. If the set~$\{x_2,x_3\}$ is mapped to~$\{p_2,p_3\}$ then the~$3$ vertices that are neither adjacent to~$x_2$ nor~$x_3$ must all be mapped to~$p_2$ or~$p_3$. In any case, two of these vertices are mapped to the same vertex in~$\{p_2,p_3\}$ and there is a vertex in~$\{x_2,x_3\}$ non-adjacent to both of them and not mapped to the same vertex in~$\{p_2,p_3\}$ contradicting the encoding.
We now assume that~$\{x_2,x_3\}$ is mapped to~$\{p_1,p_2\}$. The case where the set is mapped to~$\{x_3,x_4\}$ is symmetric.
This implies that all vertices except possibly~$x_1$ are mapped to~$p_1$ or~$p_2$. Since vertices that are mapped to~$p_2$ can have at most one neighbor mapped to~$p_1$ we conclude that~$h$ is mapped to~$p_1$ and the leaves adjacent to~$h$ are mapped to~$p_2$. Since vertices mapped to~$p_3$ cannot have two non-neighbors mapped to~$p_2$ we conclude that~$x_1$ is mapped to~$p_2$. This implies that~$x_2$ is mapped to~$p_1$ which is a contradiction since~$x_1$ cannot have two neighbors mapped to~$p_1$.
\end{proof}
}

In the appendix in Section~\ref{sec:reductions:applied}, we apply the theorem to show that various classes  (namely the classes $\forbind{2K_2 \dunion K_1,K_4}$, $\forbind{P_6,P_4\dunion P_2, K_4}$,  $\forbind{\HabcGraphFOUR{1}{0}{3}{0}, K_4}$,
bipartite~$\forbind{2P_3 \dunion K_1}$,~$\forbind{\HabcGraphFOUR{1}{0}{2}{0},K_5}$) are isomorphism complete. This in turn implies that various classes~$\forbind{H_1, H_2}$ that are superclasses of one of these are isomorphism complete (such as~$\forbind{K_3,2P_3 \dunion K_1}$).
\ifhasappendix \else
The graph classes characterized by encodings as described in \ifhasappendix{Section~\ref{sec:reductions:applied} }\else{this section }\fi contain many classes useful to show isomorphism-completeness of various hereditary graph classes. \fi These results are sufficient for the purpose of this paper, proving Theorem~\ref{thm:main:theorem}. However, we remark that there are several completeness results for classes of the form~$\forbind{H_1,H_2}$ where further restrictions on the encoded graphs become necessary. For example when excluding the diamond
, it is beneficial to require that there is only one encoded path between vertices mapped to~$p_1$ and~$p_t$.
Furthermore, there are some reductions that do not seem to fit into the reduction scheme of this section, even with further restrictions. Among those are in particular the line graph reductions (see~\cite{DBLP:conf/wg/KratschS12}).

\section{Isomorphism, Invariants and Canonical Labeling}\label{sec:invars:col:val}

\ifhasappendix \else We recall the concept of a complete invariant and the related concept of a canonical labeling.\fi
A \emph{complete graph invariant} for a graph class~$\mathcal{C}$ is a map~$\ingredient \colon \mathcal{C} \rightarrow \mathcal{D}$ into some class~$\mathcal{D}$ such that for graphs~$G_1$ and~$G_2$ in~$\mathcal{C}$ we have~$\ingredient(G_1)= \ingredient(G_2)$ if and only if~$G_1$ and~$G_2$ are isomorphic.
A \emph{canonical labeling} is a map that assigns to every graph~$G$ a graph~$C(G)$ with vertex set~$V(C(G))= \{1,\ldots,|G|\}$ and an isomorphism~$\phi \colon G \rightarrow C(G)$ such that the map assigning~$C(G)$ to~$G$ is a complete invariant.
\ifhasappendix
\else

\fi
There are relatively general techniques with which one can turn a complete invariant into a canonical labeling algorithm~\ifhasappendix\cite{DBLP:journals/eatcs/Gurevich97,DBLP:conf/csr/KoblerV08}\else\cite{DBLP:journals/eatcs/Gurevich97,DBLP:books/ws/phaunRS01/Gurevich01c,DBLP:conf/csr/KoblerV08}\fi. 
\ifhasappendix
\else
These techniques are typically only applicable when dealing with invariants for graph classes that have additional properties, for example the ability to attach gadgets to vertices in order to encode information or having balanced separators of bounded size. However, in our situation we can employ colors to achieve this. 

\fi
\ifhasappendix \else Intuitively the technique finds the first vertex of the canonical labeling by individualizing each vertex one by one and choosing the vertex which yields lexicographically the smallest invariant. By repeating the process, having already individualized the first vertex, then individualizing each possible second vertex and applying the invariant, the second vertex is found. Iterating this procedure yields the canonical labeling.\fi 

In this paper, we are mainly interested in isomorphism algorithms, as opposed to canonical labeling algorithms or complete invariants. We will therefore require a tool to simulate an invariant within one execution of our algorithm, given only an algorithm that performs isomorphism checks. \ifhasappendix \else The key difference between having an invariant and an isomorphism algorithm is that the invariant must be consistent across different executions of an algorithm.\fi Our simulated invariant will not be consistent across different calls of the same algorithm.

\begin{theorem}\label{thm:iso:algo:replaces:inv}
Let~$\mathrm{A}$ be an polynomial time algorithm with access a complete invariant~$\mathcal{O}$ for a graph class~$\mathcal{C}$ given as oracle. Suppose the outputs of~$\mathrm{A}$ are independent of the choice of the invariant~$\mathcal{O}$. If isomorphism of graphs in~$\mathcal{C}$ can be solved in polynomial time then there is a  polynomial-time algorithm~$\mathrm{B}$ whose outputs coincide with those of~$\mathrm{A}$ which does not require access to an oracle.
\end{theorem}

\begin{proof}
We simulate an algorithm that has access to a complete invariant as follows. We maintain a list of graphs that contains the first graph of each isomorphism class which has been used as input to the oracle so far. Whenever a new oracle call is made with a graph~$G$, we check, using~$Iso$, whether~$G$ is isomorphic to a graph that was already  used as input to the oracle. If so, the returned invariant is the graph in the database isomorphic to~$G$. Otherwise~$G$ is added to the list and returned as invariant. 
Concerning the running time, checking whether a graph already has an isomorphic copy in the list requires a number of calls to~$Iso$ equal to the number of graphs already added to the list. The total number of calls to~$Iso$ is thus at most quadratic in the number of oracle calls.
\end{proof}

We will use the theorem to replace an invariant with a isomorphism algorithm in the next section when dealing with modular decompositions (more precisely to prove Theorem~\ref{thm:iso:easy:when:prime:graphs:easy}).

\ifhasappendix{}\else 
The reason why the technique described in the proof of~\ref{thm:iso:algo:replaces:inv} not viable in practical applications of isomorphism algorithms is two-fold. Firstly, the algorithm suffers from a polynomial increase in particular in the space required and secondly, the technique makes the algorithm inherently sequential, whereas the advantage of canonization is its parallelizability. Another way of stating this is that the values of the invariant are not compatible across different executions of the algorithm.
\fi
\section{Colored modular decomposition}\label{sec:modular:decomp}
\differingappendixstatement{}{\section{Proofs of Section~\ref{sec:modular:decomp}}}

In this section we are concerned with modular decompositions and their application to the isomorphism problem. We will work with colored graphs since this is convenient in the graph isomorphism context. However, we will also generalize the concept of a module to that of a colored module, since this is required by our later applications. Since we do not require previous knowledge about the uncolored modular decomposition, we will not define it. We refer the reader to the survey by Habib and Paul~\cite{DBLP:journals/csr/HabibP10} for more information on the uncolored decomposition and its algorithmic applications. For the colors, we will assume that there is a linear order on the colors. In algorithmic applications such a linear order can always be obtained by comparing the bit-strings corresponding to the colors lexicographically.

\begin{definition}
Let~$G$ be a colored graph. A \emph{colored module} is a subset~$M$ of~$V(G)$ such that for all~$v\in V(G)\setminus M$,  if~$x,x'\in M$ are of the same color then either~$v$ is adjacent to both vertices~$x$ and~$x'$ or to neither~$x$ nor~$x'$.
\end{definition}

A module~$M$ is \emph{non-trivial} if it contains at least two vertices that cannot be distinguished by vertices outside of the module. That is,~$M$ is non-trivial if there are~$x,y\in M$ such that for all~$v\in V(G)\setminus M$ the vertex~$v$ is adjacent to~$x$ if and only if it is adjacent to~$y$. Note that every module that contains two vertices of the same color is non-trivial.

\begin{definition}
A map assigning to every graph~$G$ a subset of the vertices~$M(G)$ is said to be \emph{isomorphism invariant} if for every isomorphism~$\phi \colon G \rightarrow G'$ we have~$\phi(M(G)) = M(G')$. 
\ifhasappendix\else 

\fi A map that assigns to every graph~$G$ a partition of a subset of the vertices~$M(G) = \{M_1,\ldots,M_k\}$ is said to be \emph{isomorphism invariant} 
if for every isomorphism~$\phi \colon G \rightarrow G'$ we have~$M(G') = \{\phi(M'_1),\ldots,\phi(M'_{k'})\}$.
\end{definition}

\begin{definition}
A \emph{decomposition functor} is a ``map'' assigning to every graph~$G$ a partition of a subset of the vertices into modules which is isomorphism invariant.
\end{definition}
\ifhasappendix{\marginpar{\scriptsize{\vspace{-1cm}{\color{gray}{See Appendix~\ref{temoninlogy:sec} for the choice of terminology}}}}}\fi 

\hereorinappendixstatement{\label{temoninlogy:sec}
We elaborate briefly on the choice of terminology ``decomposition functor''. A decomposition functor can be seen as a covariant functor in the sense that every homomorphism between graphs induces a homomorphism between the subgraphs induced by the modules. Moreover it induces a homomorphism between the quotient graphs defined below. Similarly, color refinements can be seen as functors (see~\cite[Page 22]{SchweitzerThesis} for more information on this).
Note that a decomposition functor assigns to every graph a decomposition. It is thus not simply a decomposition of just a single graph. This is what the word functor expresses, rather than talking about just having a decomposition for graphs.
The crux is that these decompositions must be consistent across isomorphism classes. 

Technically a decomposition functor is not a set-theoretic map since the class of all graphs is not a set.  However, these category theoretic hurdles will not be important in the rest of this paper. In any case it does not seem that the word ``invariant decomposition map'' would express the concept really well.

}

A graph~$G$ is \emph{prime} with respect to a decomposition functor~$\ModDecFunc$ if~$\ModDecFunc(G)$ does not contain non-trivial modules. While there is a standard decomposition functor for the uncolored case, for the colored case it is in general not clear whether we can find a useful functor to decompose the graphs, and we have to find such a functor for a given graph class first, in order to decompose the graphs.

\ifhasappendix{\marginpar{\scriptsize{\vspace{-1cm}{\color{gray}{The appendix contains an example illustrating that this allows us to decompose graphs that are prime with respect to classical modular decomposition.}}}}}\fi
\hereorinappendixstatement{ We illustrate that colored modules allow us to decompose graphs that are prime as uncolored graphs. Consider a connected, bipartite graph with bipartition classes~$A$ and~$B$ of different colors. Let~$A'\subset A$ and~$B'\subset B$ be sets such that~$A'\cup B'$ induces a matching between the vertices in~$A'$ and the vertices in~$B'$ and such that all vertices of~$A'$ are adjacent to all vertices of~$B\setminus B'$ and all vertices of~$B'$ are adjacent to all vertices of~$A\setminus A'$. Then~$A'\cup B'$ forms a colored module, but no vertex of~$A' \cup B'$ is contained in a non-trivial classical (uncolored) module. 
}

In the remainder of this section, we argue for certain decomposition functors that graph isomorphism for a hereditary graph class can be solved in polynomial time if the isomorphism problem for colored prime graphs in the class can be solved in polynomial time. To facilitate the proof we can assume that we are given a complete invariant for the prime graphs in the hereditary graph class and then apply Theorem~\ref{thm:iso:algo:replaces:inv}.

Our next goal is to define the concept of the quotient graph. In the uncolored case, the quotient graph is obtained by replacing each module with a single vertex whose adjacency to the rest of the graph is the same as that of every vertex of the module. However, for the colored case, the adjacency to the rest of the graph depends on the color of the vertex in the module. This means that for every adjacency type we need to retain a vertex that has the same adjacency type with respect to vertices outside the module.

A \emph{replacement operator} is an isomorphism invariant map that assigns every non-trivial module~$M$ in a decomposition of a graph~$G$ an induced subgraph of~$M$ in which the vertices are possibly recolored. We require that for every adjacency type of vertices in~$M$ at least one vertex of~$M$ is maintained.
 Let~$\ingredient(M)$ be a complete graph invariant. Given a family of modules~$\{M_1,\ldots,M_k\}$ that partitions the graph, the \emph{quotient graph} is obtained by simultaneously replacing all modules using the replacement operator and then recoloring every vertex~$v$ as the triple~$(\chi(v), L,\ingredient(M_v))$, where~$\chi(v)$ is the color of~$v$ after the replacement,~$L$ is a list of the colors of vertices with the same adjacency type as~$v$, and~$\ingredient(M_v)$ is the invariant of the module containing~$v$.\ifhasappendix{
 \marginpar{\scriptsize{\vspace{-1cm}{\color{gray}{ Intuitively the goal of this coloring is to provide information how to reconstruct the original graph. However, it is not necessarily possible to retrace which vertices belong together, i.e., which vertices originated from the same module. This is resolved in the lemma below}}}}%
 \else 
 Intuitively the goal of this coloring is to provide information how to reconstruct the original graph. However \fi
 \ifhasappendix%
 \else
 
 \fi
 We say that the decomposition functor is \emph{simple} with respect to a replacement operator if for every complete invariant every quotient graph is prime. Intuitively this means that the decomposition functor provides us with maximal modules.

\ifhasappendix{
\marginpar{\scriptsize{\vspace{-1cm}{\color{gray}{\pascalsout{A more formal definition can be found in the appendix in Definition~\ref{def:quotient:graph}.}}}}
}

\fi

\begin{definition}Given a decomposition functor, we say a replacement operator is reversible if the following holds: two graphs~$G_1$ and~$G_2$ are isomorphic if and only if their colored quotient graphs with respect to the decomposition functor and the replacement operator are isomorphic.\end{definition}
Note that reversibility does not depend on the complete invariant that is used for the recoloring of the quotient graph.

 We remark that for uncolored graphs the definitions of module,  primality and the quotient graph coincide with the usual definition from the literature~(see~\cite{DBLP:journals/csr/HabibP10}). 
In  that context, the decomposition functor is typically chosen to partition the graph into components, components of the complement graph or maximal modules. The replacement operator simply replaces the entire module by one vertex. However, for the applications we have in mind, we require the more general concept of colored modules. 
There are certain situations where is it immediate that a replacement operator is reversible. 

\begin{lemma}
\ifhasappendix
A replacement operator is reversible if 
1.) There is only one trivial module, 2.)
all replacements contain only one vertex, or 3.)
non-trivial modules induce connected graphs, but the non-trivial modules are pairwise non-adjacent.
\else
A replacement operator is reversible in the following situations.
\begin{enumerate}
\item There is only one trivial module,
\item all replacements contain only one vertex, or
\item non-trivial modules induce connected graphs, but the non-trivial modules are pairwise non-adjacent.
\end{enumerate}
\fi
\end{lemma}
\proofatend 
To reconstruct the original graph we need to replace each set~$M'$ that is a replacement for module~$M$ by the original module~$M$. We also need to determine the adjacency of the vertices in~$M$ to vertices outside of~$M$. This adjacency information is encoded in the coloring of the vertices of~$M'$ by definition. However, it is not possible for all replacements to determine the sets~$M'$ that are replacements of modules. 

However, in each of the three cases in the lemma, it is possible to reconstruct the sets of vertices~$M'$ that originated from the same module. This is obvious for the first two cases. In the third case the sets~$M'$ originating form non-trivial modules are the maximal connected subgraphs consisting of vertices whose third entry in the color, which is~$\ingredient(M_v)$ by definition of the quotient graph, is the same.

The reason is that in each of these cases it is possible to reconstruct which replacement vertices originate from the same module.
\endproofatend

A reversible simple decomposition functor can be used to test isomorphism by considering only isomorphisms between quotient graphs and modules. Iterating this\ifhasappendix{ yields }\else, we will now develop an algorithm that performs\fi an isomorphism test for decomposable graphs if one has access to a complete invariant for prime graphs. %
\hereorinappendixstatement{
\begin{algorithm}[t]
\caption{$\ingredient$: A complete invariant for a hereditary class~$\mathcal{C}$ with decomposition functor.}
\label{algo:compl:invariant:moddecomp}
\begin{algorithmic}[1]
\REQUIRE $(G, \ModDecFunc, R,\ingredientprime)$: A graph~$G$ in~$\mathcal{C}$, a simple decomposition functor~$\ModDecFunc$, a reversible replacement operator~$R$, and an algorithm~$\ingredientprime$ that computes a complete invariant for graphs in~$\mathcal{C}$ which are colored and prime with respect to~$\ModDecFunc$.
\ENSURE  
$\ingredient(G)$: The value of~$G$ with respect to a complete invariant~$\ingredient$ for~$\mathcal{C}$.
\ENSUREGAP
\STATE $Q\leftarrow G$
\IF {$Q$ is not prime}
	\STATE Compute $\ModDecFunc(G)$.
	\FORALL {non-trivial $M\in \ModDecFunc(G)$}
		\IF {$M$ is prime}
			\STATE Compute~$\ingredientprime(M)$.
		\ELSE
			\STATE \label{line:recurive:call} Recursively compute~$\ingredient(M, \ModDecFunc, \ingredientprime)$.      	
		\ENDIF
	\ENDFOR
	\STATE\label{line:quotient:graph} Compute the quotient graph~$Q'$ of~$G$ with respect to~$\ModDecFunc$,~$R$, and the computed invariants~$\ingredientprime(M)$. 
	\STATE $Q \leftarrow Q'$
\ENDIF
\RETURN $\ingredientprime(Q)$ 
\end{algorithmic}
\end{algorithm}

\emph{(Description of Algorithm~\ref{algo:compl:invariant:moddecomp})} Given a reversible, simple decomposition functor~$\ModDecFunc$ and an algorithm~$\ingredientprime$ that computes a complete invariant for graphs in a graph class~$\mathcal{C}$ that are prime with respect to~$\ModDecFunc$, Algorithm~\ref{algo:compl:invariant:moddecomp} constitutes a complete invariant for graphs in~$\mathcal{C}$. The algorithm  applies the functor~$\ModDecFunc$. If this application yields non-trivial modules, the algorithm recursively computes the invariant of the modules and forms the quotient graph.   The invariant of the input is the invariant of this prime graph. 
}%
\hereorinappendixstatement{
\begin{lemma}\label{lem:iso:algo:given:invariant}
Let~$\mathcal{C}$ be a hereditary graph class and~$\ModDecFunc$ a simple polynomial time computable decomposition functor for~$\mathcal{C}$ with polynomial-time computable reversible replacement operator~$R$. Given a complete invariant~$\ingredientprime$ for colored prime graphs in~$\mathcal{C}$, Algorithm~\ref{algo:compl:invariant:moddecomp} constitutes a complete invariant for all graphs in~$\mathcal{C}$. If the total running time for query calls to~$\ingredientprime$ is polynomially bounded, then Algorithm~\ref{algo:compl:invariant:moddecomp} runs in polynomial time.
\end{lemma}
 
\begin{proof}
\emph{(Correctness)} First note that whenever a recursive call is performed in Line~\ref{line:recurive:call}, the module~$M$, which is the new input graph, has less vertices than the original graph. Thus the algorithm terminates. Suppose~$G_1$ and~$G_2$ are isomorphic graphs to which the algorithm is applied. By definition, the decomposition functor is invariant, and thus, by induction, the recursive calls compute the same invariant for corresponding modules. Since decomposition functors are isomorphism invariant, 
the quotient graphs are isomorphic. These observations 
 show that the invariants computed for~$G_1$ and~$G_2$ are the same. Suppose now that two non-isomorphic graphs~$G_1$ and~$G_2$ are given to Algorithm~\ref{algo:compl:invariant:moddecomp}.
 By induction we can assume that the invariants for the modules form a complete invariant for the modules. Therefore, the correctness of the algorithm follows from the assumption that~$R$ is reversible.

\emph{(Running time)} We show that, given a complete invariant for prime graphs  that runs in polynomial time, Algorithm~\ref{algo:compl:invariant:moddecomp} also runs in polynomial time. Consider the recursion tree of a call to Algorithm~\ref{algo:compl:invariant:moddecomp} for a graph~$G$ of size~$n$. Since all steps can be performed in polynomial time, it suffices to polynomially bound the number of nodes in the tree. The number of vertices for each node is smaller than the number of vertices of the parent. This bounds the height of the recursion tree, and thus it suffices for us to bound the number of leaves of the tree. With the possible exception of the root, for each recursive call corresponding to a leaf, the input graph is a non-trivial module~$M$ of some coloring of~$G$ and thus in particular contains two indistinguishable vertices. Since at least one of these vertices is not present in any quotient graph that has a module that contains~$M$, this pair of vertices cannot be a pair of indistinguishable vertices in any other call corresponding to a leaf. The number of leaves is thus bounded by the number of pairs of vertices.
\end{proof}
}%
Using Theorem~\ref{thm:iso:algo:replaces:inv} we can replace the requirement for a complete invariant by an isomorphism algorithm.

\begin{theorem}\label{thm:iso:easy:when:prime:graphs:easy}
Let~$\mathcal{C}$ be a hereditary graph class and~$\ModDecFunc$ a simple polynomial time computable decomposition functor with polynomial-time computable, reversible replacement operator~$R$ for colored graphs in~$\mathcal{C}$. If the isomorphism problem for colored prime graphs in~$\mathcal{C}$ can be solved in polynomial time then the isomorphism problem of all graphs in~$\mathcal{C}$ can be solved in polynomial time. 
\end{theorem}
\proofatend
The theorem follows by combining Lemma~\ref{lem:iso:algo:given:invariant} and Theorem~\ref{thm:iso:algo:replaces:inv}.
However, there is an issue with the length of the encoding of colors. In particular, if the calculated invariants are repeatedly applied to recolor the graphs then, even with a polynomial number of iterations, the encoding of the colors may become super-polynomial. To avoid this, we apply the standard technique of hashing colors. More precisely, we maintain a numbered list of colors. Every time a new color is encountered, the next unused number is assigned to this color. When an oracle call is performed by Algorithm~\ref{algo:compl:invariant:moddecomp}, and consequentially according to Theorem~\ref{thm:iso:algo:replaces:inv} the algorithm for prime graphs is invoked to simulate this oracle call, then the graph will be given in encoded form with short descriptions for the colors. Since only polynomially many colors can appear during the execution of the algorithm, this shows that the entire algorithm runs in polynomial time.
\endproofatend

The theorem in particular applies to the standard uncolored modular decomposition. This decomposition is associated with a simple decomposition functor with polynomial-time computable, reversible replacement. In his diploma thesis, Fuhlbr\"{u}ck~\cite{fuhlbrueck} also describes a reduction for the standard uncolored decomposition factor. Already when trying to solve the isomorphism problem for uncolored graphs in a graph class~$\mathcal{C}$ with the described method, the fact that the quotient graph needs to be colored implies that we require an isomorphism algorithm for colored prime graphs. 
In~\cite{rao}, Rao describes a special case of the theorem for the uncolored modular decomposition, essentially considering hereditary graph classes in which every prime graph is 
of bounded size. 
\ifhasappendix{}\else In principle the algorithm described in Rao's paper is a special case of the algorithm described here, but this is disguised slightly by the fact that in Rao's description the algorithm simultaneously performs the computation of the (uncolored) modular decomposition, the invariant for the possible prime graphs, and the hashing of colors.\fi Moreover, in the isomorphism context, the bi-join decomposition, also described in~\cite{rao}, can also be treated by using colored modules. 
\ifhasappendix\else In different terminology, the graph classes considered by Rao are sometimes called graphs of bounded composition-width or modular width.
However, we caution the reader that there are two unrelated definitions of composition width for groups and graphs, respectively, both appearing in the context of isomorphism (see~\cite{BabaiLuks}).\fi 

\ifhasappendix
\else
Together with the observation that vertices of degree~1 can be removed in an isomorphism invariant fashion by recoloring their neighbors, the previous theorem also provides means of solving isomorphism of distance hereditary graphs. Indeed, by a theorem of 
Bandelt and Mulder~\cite{Bandelt1986182} every prime distance hereditary graph has a vertex of degree at most~1. However, there is a faster, tailored algorithm for distance hereditary graphs running in linear time~\cite{uehara}.
\fi

In our applications of the Theorem~\ref{thm:iso:easy:when:prime:graphs:easy} we also use another unrelated technique of dealing with small color classes that we describe next.

\section{Bounded generalized color valence}\label{sec:gen:col:val}
\differingappendixstatement{}{\section{Proofs of Section~\ref{sec:gen:col:val}}}

In this section we show that isomorphism of graphs of bounded generalized color valence can be solved in polynomial time. However, the proofs of the lemmas and theorems within the section require familiarity with the computational group theoretic methods that have been developed within the context of the isomorphism problem. For example, when computing automorphism groups, cosets, and sets of isomorphisms between two combinatorial objects, we use a succinct representation by generators and a representative.
We also use the set stabilizer theorem~\cite{Luks:1982} for groups with composition factors of bounded size~(see also~\cite{DBLP:conf/fsttcs/ArvindDKT10,BabaiLuks,Miller2}), which implies that for a permutation group with composition factors of bounded size acting faithfully on a set, we can compute the stabilizer of any given subset in polynomial time. For a good overview over the various computational group theoretic techniques, we refer the reader to \cite{seress2003permutation}.

Before we solve their isomorphism problem, let us consider some examples of graphs of bounded generalized color valence.
Hypergraphs on sets of bounded color class size can directly be encoded as graphs by adding a vertex for every hyper edge which is adjacent to the elements of the hyper edge. Thus, a polynomial time algorithm for bounded generalized color valence also gives rise to a polynomial time algorithm for hypergraphs of bounded color class size. However, there are examples that are not captured by this. Consider a graph of maximum degree at most~$c$ consisting of a large number of components such that there is only a small number of isomorphism types among these components.
Now add an arbitrary finite number of new vertices colored with new colors such that the added vertices form color classes of size at most~$c$. The new vertices are connected via edges in an arbitrary way to the original vertices.
The resulting graph has bounded generalized color valence.
However, its automorphism group may contain composition factors of arbitrarily large size.
The isomorphism algorithm for graphs of bounded degree exploits the fact that at least for components it is possible, by individualizing one vertex, to obtain a group with bounded composition factor. However in the example graphs just described 
it is not clear how the standard group theoretic arguments can be applied.

Nevertheless, our goal is to prove that graph isomorphism of graphs with bounded generalized color valence can be solved in polynomial time. Note that the classes of bounded generalized color valence are closed under refinement operations such as individualization and naive vertex refinement. 

To solve the isomorphism problem, we need to deal with the small color classes. If there are only a bounded number of them, we can apply individualization to all vertices in small color classes. However, the number of small color classes can be linear in the number of vertices. To remedy this problem, we exploit the existence of certain colored modules and use group theoretic techniques.

We first define a decomposition functor that works well with graphs of bounded generalized color valence
that have been refined with naive vertex refinement. 
Given a subset~$S$ of vertices, we say a vertex~$v'\notin S$ contained in the color class~$C'$ is of degree dependence at most~$d$ with respect to~$S$ if
there is a vertex~$v$ in~$S$ 
such that~$v'\in N(v) \wedge |(N(v) \cap C')|\leq d$ or~$v'\notin N(v) \wedge |C' \setminus (N(v))|\leq d$. Intuitively this definition says that there is a vertex~$v\in S$ such that individualization of this vertex followed by refinement with respect to adjacency towards~$v$ produces a set of size at most~$d$ within which~$v'$ can be found.

\begin{definition}
A non-empty subset~$S$ is a~\emph{$d$-degree dependence module}, 
if 
no vertex outside~$S$ has degree dependence at most~$d$ with respect to~$S$.
\end{definition}

\begin{lemma}\label{lem:dep:modules:partition}
Let~$G$ be a graph of color valence at most~$c$ which is stable under naive vertex refinement. If~$G$ does not contain color classes of size at most~$2c$ then the minimal~$c$-degree dependence modules partition~$G$. Moreover, the map assigning every such graph the family of minimal~$c$-degree dependence modules is a polynomial time computable decomposition functor.
\end{lemma}

\proofatend
To show the first part of the lemma it suffices to show that the following holds in~$G$: if~$v$ and~$v'$ are vertices in the color classes~$C$ and~$C'$ respectively then~$|(N(v) \cap C')|\leq c$  if and only if~$|(N(v') \cap C)|\leq c$ and also~$|C' \setminus (N(v))|\leq c$ if and only if~$|C \setminus (N(v'))|\leq c$. We only show the first equivalence, the second one follows by considering the complement graph.
For the first equivalence, by symmetry, it suffices to show one direction. Thus we assume~$|(N(v) \cap C')|\leq c$. If it is not the case that~$|(N(v') \cap C)|\leq c$ then~$|C \setminus (N(v'))|\leq c$ since the graph has color valence at most~$c$. We count the number of edges and non-edges between~$C$ and~$C'$. There are at most~$|C| c$ edges and at most~$|C'| c$ non-edges. This implies that~$|C| |C'| \leq |C| c + |C'| c$ which shows that one of the color classes has size at most~$2c$.

The fact that the minimal~$c$-degree dependence modules are isomorphism invariant follows from the fact that the definition of the modules is isomorphism invariant. They form colored modules since the color valence of the graph is assumed to be at most~$c$. To compute the modules in polynomial time one can, for every vertex, compute its module by repeatedly adding vertices that have color dependence at most~$c$ with respect to the set created so far.
\endproofatend

In Lemma~\ref{lem:dep:modules:partition}, using naive vertex refinement is essential. Indeed, consider a wheel, i.e., a cycle with an added center adjacent to every other vertex. In this graph, there is only one non-trivial~$3$-degree dependence module, namely the set that contains only the center of the wheel. Thus the~$3$-degree dependence modules do not partition the graph.

\hereorinappendixstatement{
We show now that the set of isomorphisms between subgraphs induced by certain~$c$-degree dependence modules can be computed in polynomial time.

\begin{lemma}\label{lem:iso:of:dep:modules}
Let~$c$ be a positive integer.
Let~$G_1$ and~$G_2$ be colored graphs with vertices~$v_1$ and~$v_2$. For~$i \in\{1,2\}$ let~$M_i$ be a~$c$-degree dependence module that is minimal among the~$c$-degree dependence modules containing~$v_i$. Given~$M_1$ and~$M_2$, we can in polynomial time compute all isomorphisms from~$G_1[M_1]$ to~$G_2[M_2]$ (as a coset of a group).
\end{lemma}
\begin{proof}
For~$i\in \{1,2\}$, by iterating over all vertices in~$M_i$, we can find all vertices~$v_i$ such that~$M_i$ is minimal among all~$c$-degree dependence modules containing~$v_i$. Since there are only quadratically many pairs of such vertices, it thus suffices to be able to compute all isomorphisms from~$M_1$ to~$M_2$ mapping~$v_1$ to~$v_2$ assuming that~$v_1$ and~$v_2$ are given. We individualize~$v_i$ in each graph (with the same color in both graphs) and apply naive vertex refinement. We now consider the graphs~$G_1[M_1]$ and~$G_2[M_2]$. If the color classes that appear are not the same then there is no isomorphism from~$G_1[M_1]$ and~$G_2[M_2]$, so we assume otherwise. Note that we can order the color classes in a canonical way to~$C_1\ldots,C_t$ such that~$C_1 = \{v_1,v_2\}$ and such that for~$i\in \{2,\ldots,t\}$ and for vertices in~$C_i$ there is a color class~$C_j$ with~$j<i$ such that each vertex in~$C_j$ has at most~$c$ neighbors or at most~$c$ non-neighbors in~$C_i$. This enables us to compute the set of isomorphisms between the graphs induced by the first~$i$ color classes on the basis of the isomorphisms between the graphs induced by the first~$i-1$ color classes.
Indeed, this shows that the graphs have automorphism groups that have composition factors that are subgroups of~$S_c$, the symmetric group of degree~$c$.
Therefore we can apply the known techniques based on group theoretic isomorphism algorithms. 
For example we can apply techniques by Babai and Luks~\cite{BabaiLuks} or Miller~\cite{Miller2}. We refer the reader to \cite[Section 4.2]{BabaiLuks} for an explanation of all these options. 
\end{proof}
}

\hereorinappendixstatement{

We next provide a corollary of Lemma{lem:iso:of:dep:modules} that will facilitate its application in our main theorem.

\begin{corollary}\label{cor:two:modules:and:bounded:color:classes}
Let~$G_1$ and~$G_2$ be graphs of generalized color valence at most~$c$. Let~$V_1$ and~$V_2$ be the set of vertices contained in a color class of size at most~$2c$ in~$G_1$ and~$G_2$ respectively. Let~$M_1$ and~$M'_1$ be minimal~$c$-degree dependence modules of~$G_1\setminus V_1$ and let~$M_2$ be a minimal~$c$-degree dependence module of~$G_2\setminus V_2$.
We can compute in polynomial time the automorphism group of~$G[M_1\cup M_1'\cup V]$ and the automorphism group of~$G_1[M_1\cup V_1]\dunion G_2[M_2\cup V_2]$ preserving~$G_1[M_1\cup V_1]$ and~$G_2[M_2\cup V_2]$ as blocks.
\end{corollary}

\begin{proof}
The corollary follows from Lemma~\ref{lem:iso:of:dep:modules} by observing that~$M_1 \cup V_1$,~$M'_1 \cup V_1$ and~$M_2\cup V_2$ are contained in minimal~$2c$-degree dependence modules of their respective graphs.
\end{proof}
}
\hereorinappendixstatement{
In the following we will be working on two graphs simultaneously. Note that when performing naive vertex refinement on two graphs separately, it may be the case that the color degrees differ in the two graphs even within the same color class. To avoid this, we perform \emph{a joint refinement}, by performing the refinement on the disjoint union of the two graphs.

\begin{lemma}\label{lem:bijection:on:deeg:dep:modules:yields:isomorphism}
Let~$G_1$ and~$G_2$ be colored graphs, stable under the joint vertex refinement, which are of color valence at most~$c$ and which do not contain any color class of size at most~$2c$. If~$\phi$ is a map from~$V(G_1)$ to~$V(G_2)$ that maps the~$c$-degree dependence modules of~$G_1$ isomorphically to the~$c$-degree dependence modules of~$G_2$ then~$\phi$ is an isomorphism.
\end{lemma}
\begin{proof}

It suffices to show the following: if two vertices~$v$ and~$v'$ of~$G_1$ which are not contained in the same~$c$-degree dependence module are adjacent (non-adjacent) then~$\phi(v)$ and~$\phi(v')$ are adjacent (non-adjacent).
Let~$C$ be the color class of~$v'$. Since~$G_1$ has color valence at most~$c$ the vertex~$v$ has at most~$c$ neighbors in~$C$ or at most~$c$ non-neighbors in~$C$. In the former case all neighbors of~$v$ of color~$C$ are in the same~$c$-degree dependence module as~$v$. In the latter case all non-neighbors of~$v$ are in the same~$c$-degree dependence module as~$v$. This implies that all vertices of~$C$ that are not in the same~$c$-degree dependence module have the same adjacency type towards~$v$. This is also the case in~$G_2$ and this adjacency type is the same in the graphs~$G_1$ and~$G_2$, since they have been jointly refined. This shows that~$\phi$ is an isomorphism.
\end{proof}

}
\hereorinappendixstatement{
The previous proof implicitly shows that for graphs of color valence at most~$c$ that are stable under naive vertex refinement which do not contain color classes of size at most~$2c$ the~$c$-degree dependence modules are in fact colored modules in the sense of the previous section, also providing the motivation for the terminology.

A frequent observation in the graph isomorphism context is that the isomorphism problem between two graphs reduces to computation of the automorphism group of a combination of the two graphs. We will also take this route, since it drastically streamlines our proof. However, this comes at the expense of increased running time. Note that in our case, for the reduction, we cannot simply take the usual path of forming the disjoint union of two graphs and compute its automorphism group, since this union might not have bounded generalized color valence as can be seen when forming the disjoint union of two large cliques.

\begin{lemma}\label{lem:red:iso:to:aut:gen:bd:col:val}
The isomorphism problem of graphs of generalized bounded color valence at most~$c$ polynomial time reduces to the computation of the automorphism group of a graph of generalized color valence at most~$2c$.
\end{lemma}
\begin{proof}

Let~$G_1$ and~$G_2$ be colored graphs of generalized color valence at most~$c$.
To each graph~$G_i$, we add a vertex~$v_i$ adjacent to all vertices of this graph. The two added vertices obtain a color that neither appears in~$G_1$ nor on~$G_2$.
Next, we run the joint naive vertex refinement on  each of the augmented graphs to obtain the graphs~$G'_1$ and~$G'_2$. 
If the obtained graphs do not contain the same color classes with the same multiplicities we reject the graphs as non-isomorphic.

We form their disjoint union.
Let~$C$ and~$C'$ be color classes each of multiplicity more than~$2c$ in each of the graphs~$G'_1$ and~$G'_2$.
We now determine whether we should add all edges between vertices of these colors which are in different graphs or none.
If within the graphs~$G'_1$ and~$G'_2$ vertices in~$C$ have more than~$c$ neighbors in~$C'$ then we add all edges between vertices of colors~$C$ and~$C'$ if they lie in different graphs.
Note that in this case, by arguments similar to those performed in the proof of Lemma~\ref{lem:dep:modules:partition}, vertices in~$C'$ symmetrically then also have more than~$c$ neighbors in~$C$.
If this operation is performed across all pairs of color classes which are larger than~$2c$, we obtain a graph~$G$ of bounded color valence at most~$2c$. Since we added the universal vertices~$v_1$ and~$v_2$ to each of the graphs, any automorphism of the final graph must respect~$V(G_1)$ and~$V(G_2)$ as blocks.
Thus,~$G$ contains an automorphism that swaps~$v_1$ and~$v_2$ if and only if~$G_1$ and~$G_2$ are isomorphic.
\end{proof}
}

\hereorinappendixstatement{
The lemma allows us to focus on computing the automorphism group of a single graph.
\begin{theorem}
For all positive integers~$c$, the automorphism group of a graph of generalized color valence at most~$c$ can be computed in polynomial time.
\end{theorem}

\begin{proof}
 Let~$G$ be a colored graph of generalized color valence at most~$c$. Since refinement does not increase generalized color valence, we can assume that the coloring of the graph is stable under naive vertex refinement.
Let~$U$ be the set of vertices contained in a color class of size at most~$2c$.
Let~$\mathcal{V}$ be the set of minimal~$c$-degree dependence modules of~$G- U$.

Suppose~$S$ and~$S'$ are subsets of~$U$. We say two~$c$-degree dependence modules~$M$ and~$M'$ in~$ \mathcal{V}$  are equivalent relative to~$(S,S')$ if  there is an isomorphism from~$G[M\cup S\cup S']$ to~$G[M'\cup S\cup S']$ that fixes all vertices of~$S$. We write~$M\sim_{(S,S')} M'$ if~$M$ and~$M'$ are equivalent in this way and denote by~$[M]_{\sim_{(S,S')}}$ the equivalence class of~$M$. Equivalence of modules relative to~$\sim_{(S,S')}$ and thus computation of the equivalence class can be performed in polynomial time by Corollary~\ref{cor:two:modules:and:bounded:color:classes}.

Let~$S$ be a union of color classes in~$U$. By~$\Aut_S$ we denote the following automorphism group: we consider all permutations of~$S$. Such a permutation~$\phi$ is in~$\Aut_S$ if there exists a permutation~$\pi$ of the~$c$-degree dependence modules of~$G-U$ such that the following holds for all~$M\in \mathcal{V}$:
if~$\pi(M) = M'$ then there exists an isomorphism from~$G[M\cup U]$ to~$G[M'\cup U]$ which, when restricted to~$S$, agrees with~$\phi$.

We will first show how to compute~$\Aut_{U}$ and then explain how to compute the automorphism group of~$G$ from~$\Aut_{U}$. Note that~$\Aut_{\{\}}$ is the trivial group.  Consequently, to compute~$\Aut_{U}$, it suffices for us to show the following claim for an arbitrary color class~$C$ in~$U$ not contained in~$S$.

\emph{Claim.}
Given~$\Aut_S$ we can in polynomial time compute~$\Aut_{S\cup C}$.

Every element of~$\psi' \in \Aut_S$ gives rise to exactly one permutation~$\overline{\psi'}$ of the equivalence classes induced by~$\sim_{(S,U \setminus S)}$. This gives rise to an action of~$\Aut_S$ on the equivalence classes. We first compute a subgroup~$\Aut'_S$ of~$\Aut_S$ that respects a certain coloring of these equivalence classes.
For this consider the equivalence classes induced by~$\sim_{(S\cup C,U \setminus (S\cup C))}$. They form a refinement of the classes induced by~$\sim_{(S,U \setminus S)}$. We label each equivalence class~$M$ under~$\sim_{(S,U \setminus S)}$ by the multiset of block-sizes that arise when~$M$ is partitioned according to~$\sim_{(S\cup C,U \setminus (S\cup C))}$.

We now define~$\Aut'_S$ to be the subgroup of~$\Aut_S$ that stabilizes the labels. This subgroup can be computed in polynomial time since the involved groups have composition factors of bounded size (see~\cite{BabaiLuks}).

Our strategy now for computing~$\Aut_{S\cup C}$ on the basis of~$\Aut'_S$ is to find a set on which a supergroup of~$\Aut'_S$ acts such that~$\Aut_{S\cup C}$ can be observed to be a stabilizer of a computable subset.
Let~$\mathcal{M}$  be the sets of equivalence classes in~$\mathcal{V}$ with respect to~$\sim_{(S\cup C,U \setminus (S\cup C))}$.
For each~$[M]_{\sim_{(S\cup C,U \setminus (S\cup C))}}$ in~$\mathcal{M}$ we define the set of pairs~$(\phi,[M]_{\sim_{(S\cup C,U \setminus (S\cup C))}})$ where~$\phi$ is a permutation of~$C$. 
We say two pairs~$(\phi,[M]_{\sim_{(S\cup C,U \setminus (S\cup C))}})$ and~$(\phi',[M']_{\sim_{(S\cup C,U \setminus (S\cup C))}})$ are \emph{congruent} if~$|[M]_{\sim_{(S\cup C,U \setminus (S\cup C))}}| = |[M']_{\sim_{(S\cup C,U \setminus (S\cup C))}}|$ and there is an automorphism from~$G[M\cup U]$ to~$G[M'\cup U]$ 
which fixes~$S$ and is equal to ${\phi'}^{-1} \phi$ on~$C$. This definition is independent of the chosen representatives~$M$ and~$M'$ and thus congruence is well defined. Moreover, congruence is an equivalence relation. The classes of~$\mathcal{M}$ embed into the congruence classes of pairs by taking~$\phi$ to be the identity.

We define an action of the group~$Sym(C)\times \Aut'_S$ on these equivalence classes of pairs. The element~$(\psi, \psi')$ maps the congruence class of~$(\phi,[M]_{\sim_{(S\cup C,U \setminus (S\cup C))}})$ to the congruence class of~$(\psi \phi f,[M']_{\sim_{(S\cup C,U \setminus (S\cup C))}})$, where the representative~$M'$ is chosen such that~$\overline{\psi'}([M]_{\sim_{(S,U \setminus S)}}) = [M']_{\sim_{(S,U \setminus S)}}$ and~$|[M]_{\sim_{(S\cup C,U \setminus (S\cup C))}}|= |[M']_{\sim_{(S\cup C,U \setminus (S\cup C))}}|$ and~$f$ is chosen as a permutation of~$C$ such that there is an isomorphism from~$G[M\cup U]$ to~$G[M'\cup U]$ which is equal to~$f^{-1}$ and~$\psi'$ on~$C$ and~$S$, respectively.
We need to argue that this action is well defined. First, observe that since~$\Aut'_S$ is obtained by stabilizing the labeling as defined above, an~$M'$ fulfilling the requirements exists. Suppose now that~$M''$ and~$f''$ also satisfy the required properties. Then there exists an isomorphism from~$G[M' \cup U]$ to~$G[M'' \cup U]$ that is equal to~$\psi' (\psi'^{-1}) = \id$ on~$S$ and equal to~$f''^{-1}f = (\psi \phi f'')^{-1} (\psi \phi f)$ on~$C$.
Thus the pairs~$(\pi \psi \phi ,[M']_{\sim_{(S\cup C,U \setminus (S\cup C))}})$ and~$( \pi' \psi \phi,[M'']_{\sim_{(S\cup C,U \setminus (S\cup C))}})$ are congruent.
We also need to show that we have defined an action. Suppose the pair~$(\overline{\psi}, \overline{\psi'})$ maps the congruence class of $( \psi \phi f,[M']_{\sim_{(S\cup C,U \setminus (S\cup C))}})$ to the congruence class of~$(\overline{\psi}\psi \phi f \overline{f} ,[\overline{M}]_{\sim_{(S\cup C,U \setminus (S\cup C))}})$  with~$\overline{M}$ and~$\overline{\pi}$ satisfying the requirements as above. Since there is an isomorphism from~$G[M\cup U]$ to~$G[M'\cup U]$ which is equal to~$\pi\psi$ and~$\psi'$ on~$C$ and~$S$ respectively and there is an isomorphism from~$G[M'\cup U]$ to~$G[\overline{M}\cup U]$ which is equal to~$\overline{\pi}\overline{\psi}$ and~$\overline{\psi'}$ on~$C$ and~$S$ respectively, there is an isomorphism from~$G[M\cup U]$ to~$G[\overline{M}\cup U]$ which is equal to~$\overline{\pi}\overline{\psi}\pi\psi$ and~$\overline{\psi'}\psi'$ on~$C$ and~$S$ respectively. It also follows from the definition that the neutral element of~$Sym(C)\times \Aut'_S$ fixes all pairs~$(\phi,[M]_{\sim_{(S\cup C,U \setminus (S\cup C))}})$.
Thus, we have defined an action.

Abusing notation, we can consider the group of color preserving permutations of~$Sym(S\cup C)$ as a subgroup of~$Sym(S) \times Sym(C)$. In this sense, we claim that the group~$\Aut_{S\cup C}$ is the subgroup of~$Sym(C)\times \Aut'_S$ that stabilizes the image of~$\mathcal{M}$ under the above-mentioned embedding~$[M]_{\sim_{(S\cup C,U \setminus (S\cup C))}} \mapsto (\id,[M]_{\sim_{(S\cup C,U \setminus (S\cup C))}})$. Suppose first that~$(\psi,\psi')$ is an element of this stabilizer subgroup. We use the action on the congruence classes to show that~$(\psi,\psi') \in \Aut_{S\cup C}$. We define a bijection on the~$c$-degree dependence modules contained in~$[M]_{\sim_{(S\cup C,U \setminus (S\cup C))}}$ by mapping the elements of every class~$[M]_{\sim_{(S\cup C,U \setminus (S\cup C))}}$ with an arbitrary bijection to the elements of~$[M']_{\sim_{(S\cup C,U \setminus (S\cup C))}}$, where the class~$[M']_{\sim_{(S\cup C,U \setminus (S\cup C))}}$ is chosen such that the pair $(\id,[M']_{\sim_{(S\cup C,U \setminus (S\cup C))}})$ is the image of~$(\id,[M]_{\sim_{(S\cup C,U \setminus (S\cup C))}})$ under~$(\psi,\psi')$. Since~$M$ and~$M'$ are arbitrary, it suffices now to observe that there is an isomorphism from~$M$ to~$M'$ which is equal to~$\psi$ and~$\psi'$ on~$C$ and~$S$ respectively.
Conversely, if~$(\psi,\psi')$ is an element of~$\Aut_{S\cup C}$ then~$\psi'$ is an element of~$\Aut_{S}$ by definition of the groups~$\Aut_{S}$ and~$\Aut_{S\cup C}$. Since it must respect multiplicities it is thus also an element of~$\Aut'_{S}$.  
Let~$M$ be an arbitrary~$c$-degree dependence module in~$\mathcal{V}$. By definition of~$\Aut_{S\cup C}$ there is an~$M'\in  \mathcal{V}$ and an isomorphism from~$G[M\cup U]$ to~$G[M'\cup U]$ that agrees with~$(\psi,\psi')$ on~$S\cup C$. This implies that under the action of~$Sym(C)\times \Aut'_S$ the pair~$(\id,[M]_{\sim_{(S\cup C,U \setminus (S\cup C))}})$ is mapped to~$(\psi \id \psi^{-1},[M']_{\sim_{(S\cup C,U \setminus (S\cup C))}})= (\id,[M']_{\sim_{(S\cup C,U \setminus (S\cup C))}})$. This shows that elements of~$\Aut_{S\cup C}$ stabilize the image of~$\mathcal{M}$ under the embedding described above.
We want to compute this stabilizer. To achieve that the action is faithful, we additionally let the group act on the set~$A\cup C$. Since the group has bounded composition factors, the stabilizer can be computed in polynomial time. This proves the claim.

It remains to explain how to compute the automorphism group of~$G$ from the group~$\Aut_{U}$. 
Consider the homomorphism~$\pi\colon \Aut(G) \rightarrow Sym(U)$ that is obtained by restricting the permutations in~$\Aut(G)$ to the set~$U$. We first argue that~$\pi(\Aut(G)) \leq \Aut_U$. To see this consider an automorphism~$\alpha$ of~$G$. Since the~$c$-degree dependence modules of~$G-U$ are isomorphism invariant, the automorphism~$\alpha$ maps every such module~$M$ to a module~$M'$ such that there is an isomorphism from~$G[M\cup U]$ to~$G[M'\cup U]$ that is equal to~$\pi(\alpha)$ when restricted to~$U$, showing~$\pi(\Aut(G)) \leq \Aut_U$. We next argue that~$\pi$ is surjective. Let~$\psi$ be an element of~$\Aut_{U}$. By definition of~$\Aut_{U}$ there is a bijection from~$\mathcal{V}$ to~$\mathcal{V}$ such that if~$M$ is mapped to~$M'$ then there is an isomorphism from~$G[M\cup U]$ to~$G[M'\cup U]$ that is equal to~$\psi$ when restricted to~$U$. By Lemma~\ref{lem:bijection:on:deeg:dep:modules:yields:isomorphism}, combining all these mappings gives an automorphism of~$G$ that is equal to~$\psi$ when restricted to~$U$.
This shows that~$\pi$ is surjective. It also shows how to compute a lift of~$\psi$ to~$\Aut(G)$ such that this lift projects to~$\psi$ under~$\Aut(G)$. Using Corollary~\ref{cor:two:modules:and:bounded:color:classes} such a lift can be computed in polynomial time.
By the first isomorphism theorem, it now suffices to compute a lift for each generator in a generating set of~$\Aut_U$ and to compute the kernel of the map~$\pi$. 
The kernel of~$\pi$ is by definition the point-wise stabilizer of~$U$. This is equal to the automorphism group of a graph obtained from~$G$ by individualizing all elements of~$S$. Refining such a graph with naive vertex refinement yields a graph  of bounded color valence, for which we already know how to compute the automorphism group in polynomial time.
\end{proof}
}
%

By defining a faithful group action of the automorphism group of a graph on the~$c$-degree dependence modules it possible to gradually compute the automorphism group, yielding also an isomorphism test for such graphs.

\begin{theorem}\label{cor:iso:of:bd:gen:col:in:poly:time}
Graph isomorphism for colored graphs of generalized color valence at most~$c$ can be solved in polynomial time.
\end{theorem}
\proofatend
This follows directly from the previous theorem and Lemma~\ref{lem:red:iso:to:aut:gen:bd:col:val}.
\endproofatend

In  Appendices~\ref{sec:col:ref:and:double:stars}--\ref{sec:spec:bd:clique}, we apply the theorem together with Theorem~\ref{thm:iso:easy:when:prime:graphs:easy}
to show that the classes~$\forbind{K_{1,s} \dunion K_{1,s}, K_t}$,~$\forbind{P_5, K_t}$,~$\forbind{\HabcGraphFOUR{1}{0}{b}{1}, K_3}$ and~$\forbind{\HabcGraph{1}{b}{0}, K_s}$ have a polynomial-time solvable isomorphism problem. These classes include various subcases for which the complexity of the isomorphism problem was also previously not resolved.

\hereorinappendixstatement{
\section[Color refinement and forbidden double stars]{Color refinement and forbidden double stars}\label{sec:col:ref:and:double:stars}

In this section we are concerned with graphs which do not contain double stars. More precisely we consider classes of graphs that, for some positive integer~$s$, do not contain the disjoint union of two stars~$K_{1,s} \dunion K_{1,s}$ as induced subgraphs.
Our first lemma shows that bipartite graphs without double stars refine to graphs of bounded color valence under the naive vertex refinement.

\begin{theorem}\label{thm:K1s:K1s:bipartite:refinement}
In a bipartite $\forbindplain{K_{1,s} \dunion K_{1,s}}$ graph in which the vertices are colored according to a bipartition, the naive vertex refinement produces a graph of bounded color valence. Moreover, the conclusion remains true if we only forbid double stars whose centers are in the same bipartition class.
\end{theorem}
\begin{proof}
Let~$G$ be a bipartite graph without a~$K_{1,s} \dunion K_{1,s}$ subgraph that has the centers in the same bipartition class. We apply the naive vertex refinement to the graph~$G$. We now show that if~$C$ and~$C'$ are two color classes then the color valence of every vertex~$v$ in~$C$ towards vertices in~$C'$ is less than~$R = 2s$. Suppose otherwise. Let~$v$ be a vertex in~$C$ and let~$N$ be the neighbors of~$v$ in~$C'$. By our assumption~$|N| \geq R$ and~$|C'| - |N|\geq R$. 
Since the original graph was colored according to a bipartition and the new coloring is a refinement, the vertices in~$C$ are not adjacent to~$v$. Since~$G$ is~$\forbindplain{K_{1,s} \dunion K_{1,s}}$, 
the neighborhood of a vertex in~$C$ differs from the neighborhood of~$v$ by~$s$ or more vertices. Thus, there are more than~$|C| (R-s)$ edges between~$C$ and~$N$. 
This implies that every vertex of~$N$ is incident with more than~$|C|/2$ edges.
However, since neighborhoods cannot differ by more than~$s$ vertices, there are at most~$|C| s$ edges between~$C$ and~$C'\setminus N$. This implies that a vertex in~$C'\setminus N$ is incident with at most~$|C|s /(|C'|-|N|) \leq |C|s/R  \leq |C|/2$ edges, which gives a contradiction.
\end{proof}

In later sections  we repeatedly apply the theorem to show that isomorphism of graphs in various hereditary graph classes, among them for example $\forbindplain{P_5, K_t}$ graphs, can be solved in polynomial time. However, for some classes we also require 
a slight variant of Theorem~\ref{thm:K1s:K1s:bipartite:refinement} that also holds when the parts of the graphs are of bounded degree, as opposed to being independent sets.

\begin{theorem}\label{thm:K1s:K1s:individ:non:bipartite:case}
For every pair of positive integers~$s,t$ there exists integers~$c,d \in \mathbb{N}$ such that the following holds. In every~$\forbindplain{K_{1,s} \dunion K_{1,s}, K_t}$ graph, there is a set~$S$ of size at most~$d$, individualization of which, followed by naive vertex refinement, produces a graph of generalized color valence at most~$c$.
\end{theorem}
\begin{proof}
If the graph~$G$ does not contain a star of large size, then by an easy Ramsey argument, the graph has bounded degree implying the claim of the lemma (see the introduction of~\cite{DBLP:conf/wg/KratschS12} for details). Otherwise, we color the vertices of~$G$ according to adjacency type to some large star in~$G$. Note that each color class is either of bounded degree or~$\forbindplain{K_{t-1}}$. Since the number of color classes is bounded, by induction on~$t$ we thus see that there is a set~$S'$ of bounded size, such that coloring the vertices by adjacency type to~$S'$ produces color classes of bounded degree. We apply the naive vertex refinement to the graph colored in this way. 

We now show that if~$C$ and~$C'$ are two sufficiently large color classes then the color valence of every vertex~$v$ in~$C$ towards vertices in~$C'$ is bounded by some constant~$R$. Suppose otherwise. Let~$v$ be a vertex in~$C$ and let~$N$ be the neighbors of~$c$ in~$C'$. By our assumption~$|N| > R$. There are~$|C| -d -1 $ vertices in~$C$ that are not adjacent to~$v$. Since~$G$ is~$\forbindplain{K_{1,s} \dunion K_{1,s}}$, 
there are at least~$(|C| - d) (R-s)$ edges between~$C$ and~$N$. If we assume that~$R$ is larger than~$ks$ this implies that every vertex of~$N$ is incident with at least~$(|C| - d) (k-1)/k$ edges with an endpoint in~$C$.
However, there are at most~$(|C| - d) s + d R$ edges between~$C$ and~$C'\setminus N$. This implies that a vertex in~$C'\setminus N$ is incident with at most~$((|C|-d)s +d (|C'|-|N|))/(|C'|-|N|) \leq (|C|-d)s/R +d  \leq (|C|-d)/k + d$ edges. If~$k$ and~$|C|$ are sufficiently large, this gives a contradiction.
\end{proof}

Note that the theorem and its proof are similar to Theorem~\ref{thm:K1s:K1s:bipartite:refinement} and its proof. However, since we are not be able to decompose our graph into color classes that are independent sets, we can only reason about color valence between large color classes.

\begin{corollary}\label{cor:poly:double:star:Kt}
Graph isomorphism for $\forbindplain{K_{1,s} \dunion K_{1,s}, K_t}$ graphs can be solved in polynomial time.
\end{corollary}
\begin{proof}
The corollary follows by combining Theorem~\ref{thm:K1s:K1s:individ:non:bipartite:case} and Corollary~\ref{cor:iso:of:bd:gen:col:in:poly:time}.
\end{proof}

\section[P5-free graphs of bounded clique number]{\forbindONEplain{$P_5$} graphs of bounded clique number}
\label{sec:p5:bd:clique}

In this section we analyze the structure of~$\forbindplain{P_5, K_t}$ graphs developing means to solve the isomorphism problem for such a class. We will need a theorem by Bacs{\'o} and Tuza~\cite{MR1113210} that says that every connected $\forbindONE{P_5}$ graph has a dominating~$P_3$ or a dominating clique. Recall that a \emph{dominating set} is a set of vertices~$S$ such that every vertex has a neighbor in~$S$.

An instructive class of examples to have in mind are graphs obtained from disjoint unions of bipartite graphs to which a~$\forbindplain{P_5, K_{t-3}}$ graph is joined. In this class it is easy to see that even when allowing individualization of a bounded number of vertices, the naive vertex refinement does not produce graphs of bounded color valence. In the following we employ modules to circumvent this problem.

If~$G$ is a colored graph and~$S$ a set of vertices then \emph{a refinement of~$G$ by adjacency towards~$S$} is a coloring obtained by refining the coloring of~$G$ such that the vertices of~$S$ are singleton color classes and two vertices~$v$ and~$v'$ outside of~$S$ are in the same color class if and only if they have the same color in the original coloring of~$G$ and for every vertex of~$s\in S$ either both~$v$ and~$v'$ are adjacent to~$s$ or both are non-adjacent to~$s$.

\begin{lemma}\label{lem:P5:kn:most:parts:are:modules}
For all positive integers~$t$ and~$d$,
there exists a positive integer~$\ell$ such that 
refining an uncolored, connected~$\forbindplain{P_5,K_t}$ graph~$G$ by adjacency towards a dominating set of size~$d$ 
yields color classes for which the following holds: for each color~$\mathcal{C}$, among the components in the graph induced by the color class~$\mathcal{C}$, there are at most~$\ell$ components that do not form a colored module in the entire graph.
\end{lemma}
\begin{proof}
Let~$S$ be a dominating set in an uncolored, connected graph~$G$ which is~$\forbindplain{P_5,K_t}$ and let~$\mathcal{C}$ be a color class in the graph obtained from~$G$ by refining towards adjacency. Let~$C$ be a component of the graph induced by vertices of color~$\mathcal{C}$. 
If~$C$ is not a module then there exists a vertex~$v_C\notin C$ and vertices~$u_C,u'_C\in C$ such that~$v_C$ is adjacent to~$u_C$ but not adjacent to~$u'_C$. Since~$C$ is connected we can choose~$u_C$ and~$u'_C$ to be adjacent. We call the vertex~$v_C$ a distinguishing vertex for~$C$. We first claim that a vertex~$v\notin C$ cannot be a distinguishing vertex for two distinct components~$C$ and~$C'$ at the same time (i.e., if~$C\neq C'$ then~$v_C \neq v_{C'}$). Indeed, otherwise the vertices~$u_C,u'_C,v_C,u_{C'},u'_{C'}$ would form a~$P_5$.
Since there is only a bounded number of color classes (at most~$2^d+d$), it suffices to show for each color class~$\mathcal{D}$ that the number of components~$C$ in a color class~$\mathcal{C}$ that have a distinguishing vertex of color~$\mathcal{D}$ is bounded.
Since~$G$ is~$\forbindONEplain{K_t}$, it suffices to show that for distinct components~$C$ and~$C'$ within the graph induced by~$\mathcal{C}$, the vertices~$v_{C}$ and~$v_{C'}$ are adjacent if they have the same color. Suppose this is not the case. Since~$v$ and~$v'$ have the same color,~$\mathcal{D}$ say,  and the coloring is by adjacency towards~$S$, there is a vertex~$s$ in~$S$ that distinguishes vertices in~$\mathcal{D}$ from those in~$\mathcal{C}$. If~$v_{C}$ is not adjacent to~$u_{C'}$ and~$v_{C'}$
is not adjacent to~$u_{C}$ then~$v_{C},u_{C'},v_{C'}, u_{C}$ and~$s$ induce a~$P_5$, which is a contradiction. Without loss of generality we can thus assume that~$v_{C}$ is adjacent to~$u_{C'}$. Since~$v_{C}$ cannot distinguish vertices in two different components,~$v_{C}$ is also adjacent to~$u'_{C'}$. Since the set~$\{v_{C'},u_{C'}, v_{C},u_{C},u'_{C}\}$ cannot induce a~$P_5$,~$v_{C'}$ must be adjacent to a vertex in~$C$ and thus all vertices in~$C$. However, in this case~$v_{C},u'_{C'},v_{C'}, u'_{C}$ and~$s$ would induced a~$P_5$, which shows that our assumption was false and thus~$v_{C}$ and~$v_{C'}$ are adjacent.
\end{proof}

We strengthen the lemma to show that we can choose the dominating set so that the clique number in the non-module components becomes arbitrarily small.

\begin{lemma}
For~$n \in \mathbb{N} \setminus\{0\}$ and~$i\in \{1,\ldots,t\}$ there exists an integer~$d$ such that in every connected~$\forbindplain{P_5,K_t}$ graph there is a dominating set~$S$ of size at most~$d$
such that refining by adjacency towards~$S$ yields color classes for which the following holds: 
for each color~$\mathcal{C}$, every component in the graph induced by the color class~$\mathcal{C}$ is a module or~$\forbindONEplain{K_i}$.
\end{lemma}

\begin{proof}
Suppose~$n \in \mathbb{N} \setminus\{0\}$,~$i\in \{1,\ldots,t\}$ and let~$G$ be a connected graph that is~$\forbindplain{P_5,K_n}$.
We show the statement by downward induction on~$k$. For~$i= t$ there is nothing to show, since the graph~$G$ is~$\forbindONEplain{K_t}$ by assumption.
Suppose now that~$i<t$. By induction we may assume that for~$i+1$ we have already shown the statement.

Suppose~$S$ is a dominating set such that refining by adjacency towards~$S$ produces colors within which the components are modules or~$\forbindONEplain{K_{t+1}}$. By the induction hypothesis we can choose this set to be of bounded size.
By Lemma~\ref{lem:P5:kn:most:parts:are:modules} there is only a bounded number of components in the graphs induced by the color classes that do not form modules. For each such component we pick a dominating set of smallest size. By the theorem of Bacs{\'o} and Tuza~\cite{MR1113210} this dominating set has bounded size. Let~$S'$ be the union of~$S$ and all dominating sets picked for the components in the color classes. The size of~$S'$ is bounded by a function of~$t$ and~$i$. We claim that refining by adjacency produces color classes which are~$\forbindONEplain{K_i}$ or modules. To see this, note that in general when refining by adjacency towards vertices outside of a module~$M$, the set~$M$ remains a module. For all components~$C$ that were not a module, the set~$S'$ used to refine by adjacency includes a dominating set within the graph induced by~$C$. Thus, for all vertices in the new non-singleton color classes which~$C$ splits into, there is a common neighbor in the dominating set~$S'$. Thus, since the component~$C$ was previously~$\forbindONE{K_{i+1}}$, the new color classes are now~$\forbindONE{K_i}$.
\end{proof}

The lemma shows in fact that there is a set of bounded size such that the resulting components are all modules.

\begin{corollary}
For every positive integer~$t$ there exists an integer~$d$  such that in every connected~$\forbindplain{P_5,K_t}$ graph there is a dominating set~$S$ of size at most~$d$
such that refining by adjacency towards~$S$ yields color classes for which the following holds: 
for each color~$\mathcal{C}$, all components in the graph induced by the color class~$\mathcal{C}$ form a colored module in the entire graph.
\end{corollary}
\begin{proof}
Since every~$\forbindONEplain{K_1}$ graph is empty, the corollary is a reformulation of the previous lemma for the case~$i = 1$.
\end{proof}

The corollary provides us with a mechanism to decompose a~$\forbindONEplain{P_5}$  graph into modules which a have smaller clique number.

\begin{theorem}\label{cor:P5:Kn:iso:poly:time}
For every positive integer~$t$, graph isomorphism of $\forbindplain{P_5, K_t}$ graphs can be solved in polynomial time.
\end{theorem}
\begin{proof}
We show the statement by induction on~$t$. Let~$G_1$ and~$G_2$ be~$\forbindplain{P_5, K_t}$ graphs. By the previous corollary, we can guess in both graphs corresponding dominating sets~$S$ of bounded size such that refinement by adjacency produces color classes in which the components are modules. These modules are~$\forbindONE{K_{t-1}}$ since~$S$ is dominating. By induction we can compute the isomorphism type of each such module. Since the classical replacement operation is reversible,
we can replace each module by a suitably colored  single vertex and it suffices to solve isomorphism of the two quotient graphs that are obtained.
Observe that the color classes in the quotient graph are independent sets. We now argue that if~$C$ and~$C'$ are color classes in the quotient then there is no~$2K_1$ in~$C\cup C'$. Supposing otherwise, for this to happen neither~$C$ nor~$C'$ can be contained in~$S$ since vertices in~$S$ are singletons. However, since colors are obtained by refining by adjacency towards~$S$, there is a vertex~$s\in S$ that is adjacent to all vertices in one of the color classes and adjacent to no vertex of the other color class. This vertex together with the alleged copy of~$2K_1$ would form a~$P_5$ yielding a contradiction.
Consequently, by Theorem~\ref{thm:K1s:K1s:bipartite:refinement}, the naive vertex refinement refines the graphs~$G_1$ and~$G_2$ to bounded color valence.
\end{proof}

At this point we would like to caution the reader that for each particular graph class treated in the theorem, the modular decomposition 
is formed only a bounded number of times. The fact that vertices are individualized before the modules are formed prevents us from applying the  Theorem~\ref{thm:iso:easy:when:prime:graphs:easy} directly.

\section{Specific classes of triangle-free graphs}\label{sec:spec:triangle}

We now analyze certain classes of triangle-free graphs and in doing so apply the techniques developed in the previous sections of this paper. Our ultimate goal is to show that isomorphism of triangle-free~$\forbindONEplain{\HabcGraphFOUR{1}{0}{b}{1}}$ graphs (see Figure~\ref{fig:subvidided:star:H10b1}) can be solved in polynomial time. However, only little is known about isomorphism of these classes and even their structurally simpler subclasses. This forces us to use a form of bootstrapping by showing various structural properties and developing algorithms step by step.

\begin{lemma} \label{lem:bipartite:H10b0}
For every integer~$b$, the isomorphism problem of graphs that are colored, bipartite~$\forbindONEplain{\HabcGraphFOUR{1}{0}{b}{0}}$ can be solved in polynomial time.
\end{lemma}
\begin{proof}
We analyze the structure of a connected, bipartite,~$\forbindONEplain{\HabcGraphFOUR{1}{0}{b}{0}}$, colored graph~$G$. Let~$A$ and~$B$ be the bipartition classes of~$G$. We apply the naive vertex refinement to~$G$. By possibly having individualized one vertex before the refinement, we can assume that vertices in different bipartition classes have different colors. If~$G$ does not contain~$2 K_{1,2b-1}$ with centers in the same bipartition class, then by Theorem~\ref{thm:K1s:K1s:bipartite:refinement} the naive vertex refinement produces a graph of bounded color valence. Thus for such graphs, we can solve isomorphism in polynomial time. We assume therefore that in~$G$ there is an induced subgraph isomorphic to~$2 K_{1,2b-1}$ with centers in~$A$ say.

For a vertex~$v$ in~$A$ we define the set~$N_c[v]$ as the set of vertices of~$A$ obtained by exhaustively applying the following operation: we start by setting~$S= \{v\}$. If there is a vertex~$v'$ in~$A$ such that the neighborhoods of~$S$ and~$v'$ in~$B$ intersect, and the set of neighbors of~$S$ in~$B$ is not properly contained in the set of non-neighbors of~$v'$ in~$B$, then we add~$v'$ to~$S$. In this process the final set obtained is independent of the order in which the vertices are added.

This way, every vertex of~$A$ is assigned a set~$N_c[v]$. We next prove several properties that the family~$\mathcal{F} = \{N_c[v]\mid v\in A \}$ of these sets fulfills.
To do so, we first observe that, by induction, for any vertex~$x$ in~$N_c[v]$ there is a sequence~$v = v_0,v_1,\ldots,v_t=x$ with~$t\in \mathbb{N}$ of vertices in~$N_c[v]$ such that~$v_i$ and~$v_{i+1}$ have a common neighbor in~$B$. 

We claim that the sets in the family~$\mathcal{F}$ are nested, i.e., that two such sets are disjoint or one is contained in the other. Suppose~$N_c[v]$ and~$N_c[v']$ intersect but neither contains the other. Let~$v = v_0,v_1,\ldots,v_t= v'$ be a sequence in~$N_c[v] \cup N_c[v']$ such that consecutive vertices have a common neighbor in~$B$. Such a sequence must exist since the sets intersect. This implies that there is an index~$j$ such that~$v_j$ is in~$N_c[v]$ and~$v_{j+1}$ is not in~$N_c[v]$. Since~$v_j$ and~$v_{j+1}$ have a common neighbor it must be the case that the neighborhood of~$v_{j+1}$ includes all neighbors of vertices in~$N_c[v]$. Therefore~$N_c[v]$ is a subset of~$N_c[v']$, showing that~$\mathcal{F}$ is nested.

For a subset~$S$ of~$A$, as usual, we let~$N(S)$ be the set of vertices of~$B$ which are adjacent to some vertex of~$N$.
 We now argue that~$N(N_c[v])$ and~$N(N_c[v'])$ are disjoint, whenever~$N_c[v]$ and~$N_c[v']$ are disjoint. Supposing otherwise, there exists a vertex~$b \in N(N_c[v])\cap N(N_c[v'])$. This implies there are vertices~$v_1 \in N_c[v]$ and~$v_2 \in N_c[v']$ both adjacent to~$b$. If the neighborhood of~$v_1$ does not properly include all neighbors of vertices in~$N_c[v']$ then~$v_1$ should have been added to~$N_c[v']$. Otherwise all vertices of~$N_c[v']$ should have been added to~$N_c[v]$. This gives a contradiction.
 
We also claim that if~$v$ and~$v'$ are the centers in~$A$ in an induced subgraph isomorphic to~$2 K_{1,2b-1}$ then the sets~$N_c[v]$ and~$N_c[v']$ are disjoint. 
If these sets are not disjoint, there exists a sequence~$v = v_0,v_1,\ldots,v_t= v'$ in~$N_c[v]$ such that the neighborhoods of~$v_i$ and~$v_{i+1}$ intersect.
Observe that, due to the absence of the forbidden subgraph~$\HabcGraphFOUR{1}{0}{b}{0}$, the neighborhoods of two vertices in~$A$ that have a common neighbor differ by at most~$2(b-1)$ vertices, i.e., the symmetric difference of the neighborhoods contains at most~$2(b-1)$ vertices. We show by induction that~$v$ and~$v_{i}$ have a common neighbor.
For~$i=0$ we have~$v=v_0$ and there is nothing to show. Suppose we have shown the statement for~$i$. Then all but at most~$(b-1)$ neighbors of~$v$ are also neighbors of~$v_{i}$ 
and all but at most~$(b-1)$ neighbors~$v_{i}$ are neighbors of~$v_{i+1}$. Since~$v$ has more than~$2(b-1)$ neighbors, this shows that~$v_{i+1}$ and~$v$ have a common neighbor.
We conclude that~$v$ and~$v'$ have a common neighbor and thus their neighborhoods differ by at most~$2(b-1)$. This contradicts the fact that they are centers of a double star~$2 K_{1,2b-1}$. And thus shows that~$N_c[v]$ and~$N_c[v']$ are disjoint.

Let~$\mathcal{M}$ be the maximal sets~$N_c[v]$ different from~$A$ for which~$N(N_c[v])$ is 
of size at least~$2b-1$.
We claim that 
 the sets~$N_c[v] \cup N(N_c[v])$ with~$N_c[v] \in\mathcal{M}$ form modules in the graph that has been color-refined by degree. Let~$N_c[v]$ be a set in~$\mathcal{M}$. We first note that, by definition of~$N_c[v]$, every vertex of~$A$ not in~$N_c[v]$ is adjacent to all vertices or to no vertex of~$N(N_c[v])$, since otherwise the vertex would be part of~$N_c[v]$.
Suppose there is a vertex~$y$ in~$A$ not in~$N_c[v]$ that is adjacent to all vertices of~$N(N_c[v])$. This implies~$N_c[v] \subseteq N_c[y]$. We claim that every vertex~$y'$ of~$A$ not in~$N_c[v]$ with the same degree as~$y$ is also adjacent to all vertices of~$N(N_c[v])$. Indeed, otherwise they would form a double star~$2 K_{1,2b-1}$. By our observation above~$N_c[y]$ and~$ N_c[y']$ would be disjoint, which would show that~$N_c[v]$ was not maximal among the sets different from~$A$. By definition every vertex of~$B$ that is not in~$N(N_c[v])$ is not adjacent to any vertex of~$N_c[v]$. This shows that the maximal non-trivial sets are modules.

The map that associates~$\mathcal{M}$ with~$G$ is a polynomial time computable decomposition functor. By our observations above, the prime graphs do not contain~$2 K_{1,2b-1}$ with centers in~$A$ and for these we already described at the beginning of the proof how to solve isomorphism in polynomial time. The functor has a reversible replacement operator, namely maintaining two adjacent vertices of a module, an since the modules are maximal sets, the decomposition is simple.
Thus, by Theorem~\ref{thm:iso:easy:when:prime:graphs:easy} we can solve isomorphism of bipartite, colored~$\forbindONEplain{\HabcGraphFOUR{1}{0}{b}{0}}$ graphs in polynomial time.
\end{proof}
 
We can use the algorithm for the bipartite case to  develop an algorithm for all graph that~$\forbindplain{\HabcGraphFOUR{1}{0}{b}{0}}$ and do not contain triangles.

\begin{lemma}\label{lem:triangle:free:H10b0}
Isomorphism of~$\forbindplain{\HabcGraphFOUR{1}{0}{b}{0} , K_3}$ graphs can be solved in polynomial time.
\end{lemma}

\begin{proof}
Let~$G$ be a colored, connected $\forbindplain{\HabcGraphFOUR{1}{0}{b}{0} , K_3}$ graph. If the maximum degree of~$G$ is at most~$6b$
we can solve isomorphism in polynomial time. Thus, let~$v$ be a vertex of~$G$ of degree at least~$d > 6b$ 
whose neighborhood is maximal, i.e., the neighborhood of no other vertex properly contains the neighborhood of~$v$. 
Let~$N_0$ be the set of vertices whose neighborhood is the same as the neighborhood of~$v$. Let~$N_i$ be the set of vertices that have distance~$i$ from~$N_0$. Note that~$N_1$ forms an independent set since~$G$ is triangle free. 
We partition~$N_2$ into three groups depending on the number of neighbors in~$N_1$ being high, low or medium as follows: let~$N_2^h$, and~$N_2^l$ be the sets of vertices in~$N_2$ which have at most~$c_1= 3b$ non-neighbors and at most~$c_2=b$ neighbors in~$N_1$ respectively. The set~$N_2^m$  is the set of vertices in~$N_2$ neither in~$N_2^h$ nor~$N_2^l$. Since $c_1 +c_2 \leq d$ these three sets partition~$N_2$.

We first argue that the set~$N_2^h$ forms an independent set. Indeed, two vertices in~$N_2^h$ have a common neighbor in~$N_1$, since 
each such vertex has at most~$c_1$ non-neighbors in~$N_1$, but there are at least~$d > 2c_1 $ vertices in~$N_1$.

\emph{Claim.} We claim that if a vertex~$u_1$ in~$N_2$ has a neighbor~$u_2$ that is not in~$N_1$, then the set of neighbors of~$u_1$ in~$N_1$ and the set of neighbors of~$u_2$ in~$N_1$ form disjoint sets which contain all but at most~$b-1$ vertices of~$N_1$. 

The fact that the neighbor sets in~$N_1$ are disjoint follows from the triangle-freeness of the graph. If there were~$b$ or more vertices in~$N_1$ not adjacent to~$u_1$ and not adjacent to~$u_2$, we could construct the forbidden subgraph~$\HabcGraphFOUR{1}{0}{b}{0}$ by choosing the two vertices~$u_1$ and~$u_2$, a vertex in~$N_1$ adjacent to~$u_2$ (consequently not adjacent to~$u_1$), one vertex in~$N_0$ and~$b$ vertices in~$N_1$ neither adjacent to~$u_1$ nor to~$u_2$. This proves the claim.
\medskip

From the claim it follows that no vertex~$u_1$ in~$N_3\cup N_2^l$ has a neighbor~$u_2$ in~$N_2^l \cup N_2^m$. Indeed since $b+(c_1 - c_2)   \leq d$ there are at least~$b$ vertices in~$N_1$ not contained in the union of the neighborhoods of~$u_1$ and~$u_2$.

We next argue that the bipartite graph with partition classes~$N_2^h$ and~$N_2^l$ does not contain the double star~$2 K_{1,(c_1+1) b}$ with the centers in the same class. The centers of such an alleged double star~$J$ cannot be in~$N_2^h$ since such vertices have a common neighbor in~$N_1$, which, together with a subgraph of~$J$ isomorphic to~$K_{1,b}\dunion K_2$, would result in the forbidden subgraph. We now assume that the centers of~$J$ are in~$N_2^l$. Symmetrically to the other case, the centers cannot have a common neighbor in~$N_1$. 
Suppose~$w$ is one of the centers and there is a vertex~$u$ in~$N_1$ adjacent to~$w$ such that at least one vertex~$y$ of~$J$ in~$N_2^h$ is adjacent to~$u$. Then we find a forbidden subgraph with a vertex set consisting of~$y$,~$u$, the two centers and~$b$ vertices of~$J$ in~$N_2^h$ adjacent to~$w$. Supposing otherwise, let~$u$ be a neighbor of~$w$ in~$N_1$ which is not adjacent to any vertex of~$J$ in~$N_2^h$. Note that among any~$(c_1+1) b$ vertices in~$N_2^h$, there is a set~$T$ of~$b$ vertices that has a common neighbor in~$N_1$.
Thus in~$J$ we can choose a set of~$T$ vertices adjacent to a common vertex~$x$ in~$N_1$ but such that all vertices in~$T$ are not adjacent to~$u$. 
 The vertices in~$T$ together with~$x$,~$v$,~$u$ and~$w$ induce a forbidden subgraph. 
By Theorem~\ref{thm:K1s:K1s:bipartite:refinement}, the absence of double stars that we just argued implies that the naive vertex refinement produces bounded color valence in the graph induced by~$N_2^h\cup N_2^l$.

Suppose~$N_2^m$ is non-empty. We distinguish two cases according to the existence of an edge running within~$N_2^m \cup  N_2^h$.
Suppose first there exists an edge. Let~$\overline{N_2^h}$ be the set of vertices in~$N_2^h$ that have a neighbor in~$N_2^m$. We now show that~$N_2^m \cup \overline{N_2^h}$ is a complete bipartite graph. Let~$x_1$ and~$x_2$ be adjacent vertices of~$N_2^m \cup \overline{N_2^h}$. Consider a third vertex~$x_3$ in~$N_2^m$. It suffices to show that~$x_3$ is adjacent to~$x_1$ or~$x_2$. Supposing otherwise, using the claim, we see that less than~$b$ vertices of~$N_1$ are neither adjacent to~$x_1$ nor to~$x_2$. Since~$x_3 \in N_2^m$ is adjacent to at least~$c_2 =b$ vertices in~$N_1$,
this shows that~$x_3$ has a common neighbor in~$N_1$ with~$x_1$ or~$x_2$. Due to the forbidden subgraph, this implies that for~$i\in \{1,2\}$ the number of neighbors of~$x_i$ which are not neighbors of~$x_3$ is less than~$b$. This implies that the number of vertices in~$N_1$ to which~$x_3$ is not adjacent is less than~$c_1 = 3 b$, which contradicts~$x_3 \in N_2^m$.

By the claim, the~$N_1$-neighborhoods of adjacent vertices in~$N_2$ are complementary with respect to~$N_1$, with the exception of less than~$b$ vertices. Thus in the case that there is an edge~$x_1,x_2$ in~$N_2^m$ then individualizing~$x_1$ and~$x_2$ and refining by adjacency yields a graph for which every vertex has bounded color valence to color classes contained in~$N_1$.

Suppose now that there is no edge within~$N_2^m \cup  N_2^h$. Due to the forbidden subgraph, there cannot be a double star isomorphic to~$K_{1,b}\dunion K_2$ in the graph induced by~$N_1\cup N_2$ with centers in~$N_1$. We show that if there is a double star~$J$ isomorphic to $2K_{1,2b}$ between~$N_2^m$ and~$N_1$  with centers in~$N_2^m$ then the graph is not prime with respect to a certain polynomial time computable decomposition functor. Let~$x_1$ and~$x_2$ be the centers in~$N_2^m$. Suppose a vertex~$x$ in~$N_2$ has a common neighbor with both~$x_1$ and~$x_2$. Note that this vertex is neither adjacent to~$x_1$ nor to~$x_2$. This implies it must be adjacent to all but~$b-1$ neighbors of~$x_1$ and all but~$b-1$ neighbors of~$x_2$. This in turn implies that it must be adjacent to all neighbors of~$x_1$ and all neighbors of~$x_2$ since otherwise the forbidden subgraph is formed by~$x_1$,~$x_2$,~$x$, and suitable neighbors of~$x_1$ and~$x_2$. Due to the forbidden subgraph, if a vertex from~$N_2^l$ has a common neighbor with~$x_1$ then its neighborhood in~$N_1$ is a subset of the neighborhood of~$x_1$ in~$N_1$. Moreover it cannot be adjacent to a vertex in~$N_2^h$.
The argument shows that the set of vertices which have a common neighbor with~$x_1$ but not with~$x_2$, together with their neighbors in~$N_1$, form a colored module. The union of all such modules is also a module containing only vertices from~$N_1$,~$N_2^m$ and~$N_2^h$. This module is a bipartite graph which is isomorphism invariant for the graph in which~$v$ has been individualized and can be computed in polynomial time. 
By Lemma~\ref{lem:bipartite:H10b0} we can determine the isomorphism type of the module. Replacing it by an edge is reversible. Since it contains all double stars of the describes type, the associated decomposition functor is simple.
For the entire graph~$G$,  note that, when all vertices of this module except the endpoints of one edge  are removed, then the graph remains connected.

We next show for every~$i> 2$ that for every vertex~$v\in N_i$ the number of neighbors of~$v$ in~$N_i \cup  N_{i+1} $ is at most~$b-1$.
Indeed, suppose~$z_0$ is a vertex in~$N_i$ which has a set~$T$ of~$b$ neighbors in~$N_i \cup  N_{i+1}$. Using~$T$ and~$z_0$ we can construct the forbidden subgraph~$\HabcGraphFOUR{1}{0}{b}{0}$ by starting with~$z_0$ and repeatedly choosing a neighbor in the level that is one closer to the root. More precisely, let~$z_1$ be a neighbor of~$z_0$ in~$N_{i-1}$, let~$z_2$ be a neighbor of~$z_1$ in~$N_{i-2}$ and let~$z_3$ be a neighbor of~$z_2$ in~$N_{i-3}$. The set~$T\cup \{z_0,z_1,z_2,z_3\}$ induces the forbidden subgraph.

Using a very similar argument, we also show that for~$i = 2$ the analogous statement holds, i.e, for every vertex~$v\in N_2$ the number of neighbors of~$v$ in~$N_2 \cup  N_{3} $ is bounded: let~$z_0$ be a vertex in~$N_2$ which has a set~$T$ of~$b$ neighbors in~$N_2 \cup  N_{3}$. Since~$N_2$ is an independent set,~$T$ is in fact contained in~$N_3$. We construct the forbidden subgraph by adding a vertex~$z_1$ from~$N_1$ adjacent to~$z_0$, a vertex~$z_2$ from~$N_0$ and a vertex~$z_3$ from~$N_1$ not adjacent to~$z_0$. This vertex exists since, due to the neighborhood of~$v \in N_0$ being maximal, the neighborhood of~$z_0$ cannot contain all of~$N_1$.

\medskip

We summarize what we have achieved up to this point: by choosing an initial vertex~$v$ with maximal neighborhood we partition the graph into level sets~$N_i$. By either removing a bipartite module in~$N_1\cup N_2^m \cup N_2^l$ or individualizing an edge in~$N_2^m$ we obtain a colored, connected graph for which the naive vertex refinement produces a coloring in which every vertex has bounded color valence to every color class that is not contained in~$N_2^h$. 

However, between the vertices in~$N_2^h$ and~$N_3$ the same situation could arise as between~$N_1$ and~$N_2$. We cannot easily apply the previous technique to separate~$N_3$ into~$N_3^h$,~$N_3^m$ and~$N_3^l$ since the adjacency of~$N_2^h$ towards~$N_1$ is not uniform in the same way in which the adjacency of~$N_1$ is uniform to~$N_0$. 

We remedy this by reapplying the entire procedure for a bounded number of iterations using different root sets~$N_0$. Intuitively, in each iteration we force some of the vertices previously in~$N_2^h$ to be somewhere outside of what is going to be the new level set of distance 2 from the root. For this, first note that in~$N_1$ for every set~$S'$ of size~$c_1+1$, every vertex in~$N_2^h$ has a neighbor in~$S'$. For every vertex~$w\in S'$ we can choose some maximal vertex~$w'$ whose neighborhood contains the entire neighborhood of~$w$. Choosing~$w'$ and its twins to become the level set~$N_0$ implies that vertices adjacent to~$w$ will be in the new level set~$N_1$ and in particular not in the new~$N_2^h$. Repeating this for all vertices in~$S'$ ensures that for each vertex there is at least one execution during which the vertex is not contained in~$N_2^h$. If the degree of~$w'$ is not sufficiently large, i.e., it is not larger than~$6b$, we cannot apply the arguments given above. However, in this case individualizing all neighbors of~$w'$ and refining by adjacency already achieves that adjacency for all vertices towards neighbors of~$w$ is of bounded color valence.
This shows that the graph we obtain has bounded color valence. In total only a bounded number of vertices is being individualized. 

By applying the colored module isomorphism theorem (Theorem~\ref{thm:iso:easy:when:prime:graphs:easy}) this proves the lemma, since the modules are bipartite graphs (for which isomorphism can be solved by the previous lemma).
\end{proof}

The previous theorems are concerned with forbidding the graph~$\HabcGraphFOUR{1}{0}{b}{0}$ which is obtained from~$\HabcGraphFOUR{1}{0}{b}{1}$ by deleting the isolated vertex. We now study the classes obtained by deleting another vertex from~$\HabcGraphFOUR{1}{0}{b}{1}$ and forbidding the  resulting graph.
\begin{lemma}
Let~$s$ be an integer. Graph isomorphism of colored, bipartite graphs which do not contain an induced subgraph isomorphic to~$K_{1,s} \dunion K_2\dunion K_1$ such that the center of~$K_{1,s}$ and the isolated vertex are in different bipartition classes can be solved in polynomial time. \end{lemma}

\begin{proof}
If~$G$ is a connected bipartite graph then~$G$ does not contain~$K_{1,s} \dunion K_2\dunion K_1$ with the center and the isolated vertex in different bipartition classes if and only if the bipartite complement of~$G$ is~$\forbindONEplain{\HabcGraphFOUR{1}{0}{s}{0}}$.
Isomorphism of such graphs can thus by solved in polynomial time due to Lemma~\ref{lem:bipartite:H10b0} by taking the bipartite complement.
\end{proof}

\begin{corollary}\label{cor:bip:com:neigh}
Isomorphism of connected, bipartite~$\forbindONEplain{\HabcGraphFOUR{1}{0}{b}{1}}$ graphs for which every pair of vertices in the same bipartition class has a common neighbor can be solved in polynomial time.
\end{corollary}
\begin{proof}
If a connected, bipartite graph is~$\forbindONEplain{\HabcGraphFOUR{1}{0}{b}{1}}$ and every pair of vertices in the same bipartition class has a common neighbor then the graph does not contain an induced subgraph isomorphic to~$K_{1,b} \dunion K_2\dunion K_1$ such that the center of~$K_{1,b}$ and the isolated vertex are in different bipartition classes. We can thus apply the previous lemma. \end{proof}

Despite being able to solve isomorphism of bipartite~$\forbindONEplain{\HabcGraphFOUR{1}{0}{b}{1}}$ graphs with the corollary, we still require more detailed structural knowledge on these graph classes for our overall goal.

\begin{lemma}\label{lem:double:star:to:triple:star}
Let~$G$ be a bipartite graph colored according to a bipartition which does not contain an induced subgraph isomorphic to~$K_{1,s} \dunion K_2\dunion K_1$ such that the center of~$K_{1,s}$ and the isolated vertex are in different bipartition classes. If~$G$ does not contain the bipartite complement of~$3 K_{1,s}$ with centers in the same bipartition class, then by individualizing a set of bounded size and applying the naive vertex refinement, the graph can be turned into a graph of bounded color valence. 
\end{lemma}
\begin{proof}
If the graph has no $2 K_{1,s}$ with centers on the same side, then by Theorem~\ref{thm:K1s:K1s:bipartite:refinement} the naive vertex refinement produces a graph of bounded color valence. Otherwise, we individualize the vertices of a copy~$J_1$
 of $2 K_{1,s}$ with centers on the same side.
If in the new color classes there is no $2 K_{1,s}$ with centers on the same side, we can apply Theorem~\ref{thm:K1s:K1s:bipartite:refinement} again. Otherwise, we individualize one new copy~$J_2$ of $2 K_{1,s}$ such that prior to individualization vertices in~$J_2$ that lie in the same bipartition class had the same color. Note that the vertex sets of~$J_1$ and~$J_2$ are disjoint. By possibly repeating the argument, we conclude that we can choose~$J_1$ and~$J_2$ such that all centers are in the same bipartition class. 
Since~$J_2$ was chosen after the vertices were refined by adjacency to~$J_1$, to every vertex of~$J_1$ the centers of~$J_2$ are either both adjacent or neither of them is adjacent. This implies that both centers of~$J_2$ must be adjacent to the non-centers of~$J_1$, otherwise we find the forbidden subgraph since the centers of~$J_2$ would have a common non-neighbor in the other bipartition class. Similarly, there cannot be a non-center of~$J_2$ to which both centers of~$J_1$ are non-adjacent. This implies that one of the centers of~$J_1$ is adjacent to all non-centers of~$J_2$. This center and its non-adjacent non-center vertices in~$J_1$ together with the vertices of~$J_2$ form the bipartite complement of~$3 K_{1,s}$.
\end{proof}

We assemble the structural analysis of the subcases performed in this section to develop an algorithm deciding isomorphism of triangle-free~$\forbindONEplain{\HabcGraphFOUR{1}{0}{b}{1}}$ graphs.

\begin{theorem}\label{thm:K3:H10b1:poly:time}
The isomorphism problem for $\forbindplain{\HabcGraphFOUR{1}{0}{b}{1}, K_3}$ graphs can be solved in polynomial time.
\end{theorem}

\begin{proof} Without loss of generality we assume~$b \geq 3$.
If one of the input graphs to a given instance of the graph isomorphism problem is $\forbindplain{\HabcGraphFOUR{1}{0}{b}{0}, K_3}$, then we can solve isomorphism using Lemma~\ref{lem:triangle:free:H10b0}.

We thus discuss the structure of a~$\forbindplain{\HabcGraphFOUR{1}{0}{b}{1}, K_3}$ graph~$G$ that contains an induced subgraph~$K$ isomorphic to~$\HabcGraphFOUR{1}{0}{b}{0}$. 
Note that every vertex of~$G$ is adjacent to at least one vertex of~$K$, i.e.,~$K$ is dominating. For~$i\in \{0,\ldots,4\}$ we define the sets~$N_i$ 
to be set of vertices of distance~$i$ from the leaf in~$K$ that is adjacent to a vertex of degree 2. By individualizing the vertices in~$K$ and recoloring the vertices of~$G$ according to adjacency towards~$K$ we obtain a coloring of the graph into independent sets.

Let~$C$ and~$C'$ be two color classes which contain adjacent vertices. Let~$T$ be the set of neighbors in~$K$ of a vertex (and thus any vertex) in~$C$ and let~$T'$ be the set of neighbors in~$K$ of a vertex in~$C'$. This implies that~$T$ and~$T'$ form disjoint non-empty independent sets of~$K$. 
We want to achieve that the graph induced by two color classes~$C$ and~$C'$ has bounded color valence when refined with the naive vertex refinement.
Note that the subgraph induced by~$C$ and~$C'$ does not contain~$K_{1,2b+1} \dunion K_2\dunion K_1$ such that the center of~$K_{1,2b+1}$ and the isolated vertex are in different bipartition classes: Otherwise, by adding an appropriate vertex of~$K$, we could construct the forbidden subgraph~$\HabcGraphFOUR{1}{0}{b}{1}$. Thus, if the graph induced by the two colors does not contain the bipartite of complement~$3 K_{1,2b+1}$ with centers in the same bipartition class, we can, by Lemma~\ref{lem:double:star:to:triple:star}, individualize a set of bounded size such that the refinement yields bounded color valence. Since the bipartite complement of~$3 K_{1,2b+1}$ in particular contains the graph~$\HabcGraphFOUR{1}{0}{b}{0}$, the graph induced by~$C$ and~$C'$ can only contain the bipartite complement of~$3 K_{1,2b+1}$ if~$T \cup T'$ contains every vertex of~$K$. 

In the following we thus assume that one of the sets~$T$ and~$T'$ is~$N_0\cup N_2\cup N_4$ while the other one is~$N_1\cup N_3$. We also assume that the graph induced by~$C$ and~$C'$ contains the bipartite of complement~$3 K_{1,2b+1}$ with centers in the same bipartition class. Let~$J$ be an isomorphic copy of this graph. Without loss of generality we assume that the centers are in the color class~$C$.

We claim that~$G$ is a bipartite graph such that each pair of vertices in one of the bipartition classes has a common neighbor. To show this, we first show that for each vertex~$v$ not in~$C$, for two of the three stars in~$J$,~$v$ is adjacent to more than half the vertices of color~$C$ or to more than half the vertices of color~$C'$. Thus let~$v$ be a vertex that is not of color~$C$, its color being~$C''$ say. Let~$T_v$ be the vertices in~$K$ to which~$v$ is adjacent. Recall  that~$T$ is a maximal independent set of~$K$. Thus, since~$T_v$ is different from~$T$, there is a vertex~$x$ in~$K$ such that~$x$ is adjacent to the vertices in~$C$ but neither adjacent to the vertices in~$C'$ nor the vertices in~$C''$. This implies, due to forbidden subgraph~$\HabcGraphFOUR{1}{0}{b}{1}$, that~$v$ is adjacent to at least two centers of~$J$ or for two of the stars in~$J$,~$v$ is adjacent to more than half of the vertices in~$C'$.
Since~$T_v$ is non-empty, there is no edge between~$v$ and~$C$ or there is no edge between~$v$ and~$C'$.

Let~$A$ be the set of vertices that have a neighbor in~$C$ and let~$B$ be the other vertices. We claim that this is a bipartition of~$G$. Indeed, every vertex that has a neighbor in~$C$ is adjacent to two of the tree centers. Thus all vertices in~$A$ have a common neighbor and~$A$ is an independent set. All vertices in~$B$ which are not in~$C$ are adjacent to more than half of the vertices of color~$C'$ in two of the stars. It remains to show that vertices in~$C$ have a common neighbor with all vertices in~$B$. However, each vertex in~$B$ is adjacent to a vertex in~$T$ which in turn is adjacent to every vertex in~$C$. Thus vertices in~$B$ have a common neighbor and~$B$ is an independent set.

This proves that~$G$ is bipartite such that each pair of vertices in one of the bipartition classes has a common neighbor. Isomorphism of such graphs can be solved using Corollary~\ref{cor:bip:com:neigh}.
\end{proof}

\section{Specific classes of bounded clique number}\label{sec:spec:bd:clique}

In this section we show that isomorphism of $\forbindplain{\HabcGraph{1}{b}{0}, K_s}$ graphs can be solved in polynomial time (see Figure~\ref{fig:subvidided:star:H1b0}). To do so, we start with a structural lemma concerning such graphs.

\begin{lemma}\label{lem:cannnot:connect:double:claws}
Let~$G$ be an~$\forbindONEplain{\HabcGraph{1}{b}{0}}$ graph and let~$H$ be an induced subgraph of~$G$ isomorphic to~$K_{1,2b-1}$. Let~$M$ be a connected induced subgraph of~$G$ such that~$M$ and~$H$ are disjoint and non-adjacent.
No vertex of~$G$ simultaneously has a neighbor in~$H$ and in~$M$ and a non-neighbor in~$M$.

\end{lemma}

\begin{proof}
Suppose otherwise that~$v$ has a neighbor in~$H$ and~$M$ and a non-neighbor in~$M$. Since~$M$ is connected,~$v$ has a neighbor~$x$ in~$M$ and a non-neighbor~$y$ in~$M$ such that~$x$ and~$y$ are adjacent. If~$v$ is adjacent to at least~$b$ leaf vertices of~$H$ then~$\{y,x,v\}$ together with these~$b$ leaves induces the forbidden subgraph. Otherwise~$v$ is non-adjacent to~$b$ leaves of~$H$. If~$v$ is adjacent to the center~$h$ of~$H$ then the set~$\{x,v,h\}$ together with~$b$ leaves non-adjacent to~$v$ induce a forbidden subgraph. If however~$v$ is non-adjacent to the center, then by assumption there is a leaf~$\ell$ to which~$v$ is adjacent. In this case~$v,\ell,h$ and~$b$ leaves non-adjacent to~$v$ induce a forbidden subgraph.
\end{proof}

\begin{theorem}\label{thm:poly:for:H1b0:Ks}
Graph isomorphism for colored $\forbindplain{\HabcGraph{1}{b}{0}, K_s}$ graphs can be solved in polynomial time.
\end{theorem}

\begin{proof}
Let~$G_1$ and~$G_2$ be~$\forbindplain{\HabcGraph{1}{b}{0}, K_s}$ graphs. 
Recall that by Corollary~\ref{cor:poly:double:star:Kt} we can solve isomorphism of~$G_1$ and~$G_2$ in polynomial time if one of the graphs is~$\forbindplain{K_{1,2b-1} \dunion K_{1,2b-1}}$.

We now argue that for graphs that do contain~$K_{1,2b-1} \dunion K_{1,2b-1}$, we can use the standard modular decomposition functor for uncolored graphs to decompose the graphs. That is, if the graphs are not connected, we take the connected components as the modules. Analogously, if their complements are not connected we take the connected components of the complements as modules. 
In any other case, we decompose the graph into the maximal modules. These partition the vertices and can be computed in polynomial time (see~\cite{DBLP:journals/csr/HabibP10}). We show that if a~$\forbindONEplain{\HabcGraph{1}{b}{0}}$ graph contains an induced double star~$\forbind{K_{1,2b-1} \dunion K_{1,2b-1}}$ then there is a non-trivial module. For this it suffices to show that the minimal modules containing each~$K_{1,2b-1}$ are disjoint. Call the stars~$H_1$ and~$H_2$. Suppose the minimal module containing~$H_1$ contains a vertex from~$H_2$. The minimal module containing~$H_1$ is obtained by starting with~$S_0 = V(H_1)$ and repeatedly increasing this set by adding vertices. More precisely~$S_{i+1} =  S_{i} \cup \{v\}$ for some vertex~$v$ that has a neighbor and a non-neighbor in~$S_i$. Note that by induction all~$S_i$ are connected. By our assumptions, in some step a vertex must be added which is not in~$V(H_2)$ but has a neighbor in~$H_2$. Let~$i$ be the first step in which this happens and suppose~$u$ is the vertex that is added. Then~$u$ has a neighbor and a non-neighbor in the set~$S_i$ which induces a connected graph and a neighbor in~$H_2$, which is a contradiction to Lemma~\ref{lem:cannnot:connect:double:claws}.

Since we can solve isomorphism of colored prime graphs in polynomial time, this concludes the proof using  Theorem~\ref{thm:iso:easy:when:prime:graphs:easy} with the classical decomposition functor.
\end{proof}
}

\section{Comprehensiveness of the case distinction}\label{sec:comprehensive}

In this section we argue that the theorems developed throughout this paper together with the theorems  from~\cite{DBLP:conf/wg/KratschS12}, with the exception of finitely many cases, resolve the complexity of graph isomorphism for~$\forbindplain{H_1,H_2}$ graphs. In fact, said theorems also either provide a polynomial time algorithm or a reduction from the general isomorphism problem to the respective classes.

\begin{proof}[Proof of Theorem \ref{thm:main:theorem}]

Let~$\forbind{H_1,H_2}$ be the graph class for which we want to determine the complexity. By the results in~\cite{DBLP:conf/wg/KratschS12}, we may assume that one of the graphs,~$H_2$ say, is a complete graph. Since graph isomorphism for~$\forbindONEplain{K_2}$ graphs is polynomial-time solvable, we further assume that~$H_2$ has at least 3 vertices. By~\cite[Lemma 2]{DBLP:conf/wg/KratschS12} we may consequently assume that~$H_1$ is a forest of subdivided stars, since the problem is graph isomorphism complete otherwise.

If~$H_1$ contains 3 non-trivial components, then~$H_1$ contains~$3K_2$ and the problem is isomorphism complete~\cite[Lemma 5]{DBLP:conf/wg/KratschS12}.
In the following we can thus assume that~$H_1$ has at most~2 non-trivial components.
If~$H_1$ does not contain a vertex of degree 3 then if~$H_1$ is sufficiently large, it is either of the form~$P_i \dunion I_t$ with~$i\leq 3$
or it contains the graph~$2K_2 \dunion I_2$. In the former case graph isomorphism is polynomial-time solvable by Theorem~\ref{thm:K3:H10b1:poly:time}. In the latter case the isomorphism problem is graph isomorphism complete~\cite[Lemma 5]{DBLP:conf/wg/KratschS12}.

\emph{(The case $H_2 = K_3$).} Suppose~$H_2$ is the graph~$K_3$. Since~$H_2$ is fixed, we can assume that~$H_1$ is sufficiently large and we can thus assume by the observation above that~$H_1$ contains a vertex of degree 3.
If~$H_1$ does not contain a~$P_4$ then~$H_1$ is a union of at most two stars plus isolated vertices. If there is at most one star, the problem is polynomial-time solvable by  Theorem~\ref{thm:K3:H10b1:poly:time} (or by~\cite[Theorem 4]{DBLP:conf/wg/KratschS12}). Assuming there are two stars, if there is more than one isolated vertex then~$H_1$ contains~$2 K_2 \dunion 2K_1$ and the problem is isomorphism complete by~\cite[Lemma 5]{DBLP:conf/wg/KratschS12}. If neither of the stars is only a single edge then if there exists an isolated vertex in~$H_1$ the problem is isomorphism complete by Theorem~\ref{thm:bipart:2P3:K1:free:iso:compl} and if there is no isolated vertex in~$H_1$ it is polynomial-time solvable by Corollary~\ref{cor:poly:double:star:Kt}.
Finally, if one of the stars is only an edge then by Theorem~\ref{thm:K3:H10b1:poly:time} the problem is in polynomial-time solvable.

If~$H_1$ contains a~$P_4$ and there are two non-adjacent vertices not in the same connected component as the~$P_4$ then~$H_1$ contains~$P_4\dunion 2K_1$ and the problem is graph isomorphism complete~\cite[Lemma 5]{DBLP:conf/wg/KratschS12}. Assuming this is not the case implies that the vertex of degree 3 is in the same component as the~$P_4$. Since isomorphism of~$(2K_2 \dunion 2K_1)$-free triangle-free graphs is isomorphism complete~\cite[Lemma 5]{DBLP:conf/wg/KratschS12}, by assuming that~$H_1$ is sufficiently large, we may further assume there is at most one additional vertex not in the connected component of the~$P_4$. If there is only one degree~1 vertex non-adjacent to the vertex~$h$ of degree at least 3, 
and this vertex has distance~$2$ from~$h$, then~$H_1$ is an induced subgraph of~$\HabcGraphFOUR{1}{0}{b}{1}$ for some positive integer~$b$ and the problem is in polynomial-time solvable by Theorem~\ref{thm:K3:H10b1:poly:time}.
If~$H_1$ contains two leaves of distance at least 3 from the center, then if~$H_1$ is sufficiently large it also contains~$2K_2 \dunion 2K_1$. If~$H_1$ contains a leaf of distance at least 4 then~$H_1$ contains~$P_4\dunion 2K_2$. In both cases the problem is graph isomorphism complete by~\cite[Lemma 5]{DBLP:conf/wg/KratschS12}.

\emph{(The case $H_2 = K_n$ with~$n>3$).} Suppose now~$H_2$ is the graph~$K_n$ for some~$n\geq 4$.
Suppose first,~$H_1$ contains two non-trivial components. If it contains an isolated vertex, then in contains~$2K_2 \dunion K_1$ and the problem is isomorphism complete by Theorem~\ref{thm:K4:and:various:free:iso:compl}. If one of the components contains~$P_4$ and the graph is not connected, then the graph contains~$P_4\dunion K_1$ and the problem is isomorphism complete by~\cite[Theorem 3]{DBLP:conf/wg/KratschS12}. Otherwise the two components form a double star and the problem is polynomial-time solvable by Corollary~\ref{cor:poly:double:star:Kt}.

Thus we may assume now that there is only one non-trivial component. If there is no~$P_4$ in~$H_1$, the problem is solvable by~\cite[Theorem 4]{DBLP:conf/wg/KratschS12}. Otherwise, by~\cite[Theorem 3]{DBLP:conf/wg/KratschS12}, we can assume that~$H_1$ is connected.
If~$H_1$ is isomorphic to~$P_5$, the problem is solvable by Theorem~\ref{cor:P5:Kn:iso:poly:time}.
If~$H_1$ is isomorphic to~$P_6$ the problem is isomorphism complete by Theorem~\ref{thm:K4:and:various:free:iso:compl}. We can thus assume that~$H_1$ contains a vertex of degree at least 3, which we call the center. If there is only one leaf not adjacent to the center and this leaf has distance 2 from the center, then the problem is polynomial-time solvable by Theorem~\ref{thm:poly:for:H1b0:Ks}. If there are two leaves not adjacent to the center then~$H_1$ contains~$P_4\dunion K_1$ and the problem is isomorphism complete by~\cite[Theorem 3]{DBLP:conf/wg/KratschS12}.

If  there is a leaf of distance at least 3 from the center and the center has degree 4 then the problem is graph isomorphism complete by Theorem~\ref{thm:K4:and:various:free:iso:compl}. We can thus assume that the center has degree 3. 
Under these conditions, if~$H_2 = K_4$ and~$H_1$ is sufficiently large, then~$H_1$ contains~$P_6$ and the problem is isomorphism complete by Theorem~\ref{thm:K4:and:various:free:iso:compl}.
In the remaining cases~$H_2 = K_n$ with~$n\geq 5$ and~$H_1$ has a leaf of distance at least~$2$ from the center. In this case the problem is isomorphism complete by Theorem~\ref{thm:H1020:K5:free:iso:compl}. 
\end{proof}

\section{Conclusion}\label{sec:conclusion}

There is an intricate relationship between the boundedness of the clique width on a graph class and polynomial-time solvability of the isomorphism problem. 
While there are classes of unbounded clique width the isomorphism problem of which is solvable in polynomial time (for example the graphs of bounded degree or the graphs characterized by a forbidden star and a forbidden clique (see~\cite{DBLP:conf/wg/KratschS12})), there are no classes known to be isomorphism-complete and known to have bounded clique width. It is conceivable, but open, that isomorphism on graphs of bounded clique-width is solvable in polynomial time. In fact, it seems that typically an isomorphism reduction to a graph class~$\mathcal{C}$ yields a proof for the unboundedness of the clique width. 
It is easy to see that every class of graphs encoded by a simple path encoding has unbounded clique width. The reason is that if~$G'$ is the graph produced from a graph~$G$ by the reduction detailed in the proof of Theorem~\ref{thm:simple:path:enc} then~$G$ can be obtained from~$G'$ by a finite application of edge complementations between two sets, followed by an unbounded number  of  local complementations and vertex removals. This shows that the clique width of~$G$ is bounded by a function of the clique width of~$G'$. Thus, if~$G$ has large clique width then~$G'$ has large clique width.
\begin{corollary}
The graphs encoded by a simple path encoding have unbounded clique width.
\end{corollary}
This implies that many classes~$\forbind{H_1,H_2}$ have unbounded clique width. In particular the class~$\forbind{P_4 \cup K_1,K_4}$ has unbounded clique width, which
is 
an open case in the list of~\cite{DBLP:journals/corr/DabrowskiP14a}. 
Conversely, when a graph class is shown to be polynomial-time solvable, often this is due to some structural insight that actually amounts to showing that the clique width of the class is bounded. On another note it appears that the solution of the generalized color valence problem presented in this paper, more precisely a generalization of the technique, is a first step towards showing polynomial-time solvability of graphs of bounded clique width. More concretely, the generalization applies to graphs with fixed rank decomposition.

As mentioned in the introduction, for various other computational problems, classification results have been considered for classes characterized by two forbidden graphs. For these problems, there had been extensive prior work and  algorithmic techniques were already available. For the isomorphism problem, such techniques were lacking and the intention of this paper is to develop them. The fact that for new classes the complexity of isomorphism can be determined using these techniques, while other methods seem to fail, shows that these techniques provide something conceptually different. 

However, concerning a list of open cases that still remain for the classes characterized by two forbidden subgraphs one has to analyze precisely to which classes the techniques can be applied. For example, Brandst{\"a}dt and Kratsch~\cite{DBLP:journals/dam/BrandstadtK05} show that cycles of length~5 in~$\forbindplain{P_5,\overline{P_4\cup K_1)}}$ graphs are disjoint. By coloring vertices inside a 5-cycle depending on whether they have a neighbor outside the cycle, we can construct a decomposition functor. Furthermore, there is a description of the relevant prime graphs in~\cite{DBLP:journals/dam/BrandstadtK05} which allows to us apply the modular decomposition technique. This is just an example and a comprehensive description of the various cases that can actually be handled with the techniques 
remains as future work.

\bibliographystyle{abbrv}
\bibliography{main}

\ifhasappendix
\newpage
\begin{appendix}
\noindent{\huge{\textbf{Appendix}}}

\medskip 

The appendix contains proofs of various lemmas and theorems from the main text. It also contains lemmas required for these proofs and the proofs of these lemmas. Moreover it shows how the techniques presented in the paper are applied and how they can be used to resolve the complexity of all but finitely many classes characterized by two forbidden induced subgraphs proving Theorem~\ref{thm:main:theorem}.


\printproofs
\end{appendix}
\bibliographystyleapp{abbrv}
\bibliographyapp{main}
  \tableofcontents

\fi

\end{document}